\documentclass[11pt]{article}
\usepackage{complexity}

\usepackage[english]{babel}
\usepackage[utf8x]{inputenc}
\usepackage[T1]{fontenc}
\usepackage{enumitem}
\usepackage[margin=1.in,marginparwidth=1.75cm]{geometry}

\usepackage[algo2e, ruled, vlined]{algorithm2e}
\usepackage{amsmath}
\usepackage{amsfonts}
\usepackage{graphicx}
\usepackage{amsthm}
\usepackage{amssymb}
\usepackage{dsfont}
\usepackage[algo2e, ruled, vlined]{algorithm2e}
\usepackage[noend]{algorithmic} 
\usepackage{bbm}

\usepackage[T1]{fontenc}
\usepackage{microtype}          %
\usepackage[full]{textcomp}

\usepackage{lipsum}
\usepackage{xspace}
\usepackage{relsize}
\usepackage{framed}
\usepackage{xcolor}
\usepackage{graphicx}
\usepackage{multirow}
\usepackage{subcaption}
\usepackage{float}

\usepackage[colorinlistoftodos]{todonotes}
\usepackage[colorlinks=true, allcolors=blue]{hyperref}

\newtheorem{theorem}{Theorem}

\newtheorem{lemma}{Lemma}
\newtheorem{corollary}{Corollary}

\newtheorem{remark}{Remark}

\theoremstyle{definition}
\newtheorem{definition}{Definition}

\newcommand{\calN}{{\mathcal{N}}}

\newcommand{\good}{\mathsf{good}}

\newcommand{\knote}[1]{{\bf{\color{blue}[\tiny Karthik: #1]}}}
\newcommand{\cnote}[1]{{\bf{\color{magenta}[Chandra: #1]}}}

\usepackage{float}
\usepackage{booktabs}
\usepackage{multirow}
\usepackage{thmtools} 
\usepackage{thm-restate}
\usepackage{lipsum}

\newcommand{\Z}{\mathbb{Z}}

\newif\ifdraft
\draftfalse

\newcommand{\DPB}{\textsc{Disj-Polymatroid-Bases}\xspace}
\newcommand{\ODPB}{\textsc{Online-Disj-Polymatroid-Bases}\xspace}
\newcommand{\DSC}{\textsc{Disj-Set-Cover}\xspace}
\newcommand{\ODSC}{\textsc{Online-Disj-Set-Cover}\xspace}
\newcommand{\DMB}{\textsc{Disj-Matroid-Bases}\xspace}
\newcommand{\ODMB}{\textsc{Online-Disj-Matroid-Bases}\xspace}
\newcommand{\MatrixDB}{\textsc{Disj-Matrix-Bases}\xspace}
\newcommand{\OMatrixDB}{\textsc{Online-Disj-Matrix-Bases}\xspace}
\newcommand{\DST}{\textsc{Disj-Spanning-Trees}\xspace}
\newcommand{\ODST}{\textsc{Online-Disj-Spanning-Trees}\xspace}
\newcommand{\DCSS}{\textsc{Disj-Conn-Spanning-Subhypergraphs}\xspace}
\newcommand{\ODCSS}{\textsc{Online-Disj-Conn-Spanning-Subhypergraphs}\xspace}
\newcommand{\opt}{\texttt{opt}}

\title{Online Disjoint Spanning Trees and Polymatroid Bases\thanks{Univ. of Illinois, Urbana-Champaign, Urbana, IL 61801. Email: {\tt
      \{karthe, chekuri, weihaoz3\}@illinois.edu}. Supported in part by NSF grant CCF-2402667.}}
\author{Karthekeyan Chandrasekaran
\and Chandra Chekuri
\and Weihao Zhu 
}

\begin{document}

\maketitle

\pagenumbering{gobble}
\begin{abstract}
Finding the maximum number of disjoint spanning trees in a given graph is a well-studied problem with several applications and connections. The Tutte-Nash-Williams theorem provides a min-max relation for this problem which also extends to disjoint bases in a matroid and leads to efficient algorithms \cite{Schrijver-book}. Several other packing problems such as element disjoint Steiner trees, disjoint set covers,  and disjoint dominating sets are NP-Hard but admit an $O(\log{n})$-approximation \cite{FHKS02,CS07}. C{\u{a}}linescu, Chekuri, and Vondr\'{a}k \cite{CCV09} viewed all these packing problems as packing bases of a polymatroid and provided a unified perspective. 
Motivated by applications in wireless networks, recent works have studied the problem of packing set covers in the \emph{online} model \cite{PBV15, EGK19, BBJ24}. 
The online model poses new challenges for packing problems. In particular, it is not clear how to pack a maximum number of disjoint spanning trees in a graph when edges arrive online.  Motivated by these applications and theoretical considerations we formulate an online model for packing bases of a polymatroid, and describe a randomized algorithm with a polylogarithmic competitive ratio. Our algorithm is based on interesting connections to the notion of quotients of a polymatroid that has recently seen applications in polymatroid sparsification \cite{quanrud2024quotient}. We generalize the previously known result for the online disjoint set cover problem \cite{EGK19} and also address several other packing problems in a unified fashion. 
For the special case of packing disjoint spanning trees in a graph (or a hypergraph) whose edges arrive online, we provide an alternative to our general algorithm that is simpler and faster while achieving the same poly-logarithmic competitive ratio. 
\end{abstract}

\newpage
\pagenumbering{arabic}

\section{Introduction}\label{section:intro}
Finding the maximum number of disjoint spanning trees in a given graph is a well-studied problem with several applications and connections. The Tutte-Nash-Williams theorem provides a min-max relation for this problem which also extends to the maximum number of disjoint bases in a matroid with efficient algorithms \cite{Schrijver-book}. Several other packing problems such as disjoint set covers, disjoint dominating sets, and element disjoint Steiner trees are NP-Hard but admit an $O(\log{n})$-approximation\footnote{We are using the convention of $\alpha$-approximation for a maximization problem with $\alpha > 1$ with the understanding that the returned value is at least $\opt/\alpha$.} \cite{FHKS02,CS07}. C{\u{a}}linescu, Chekuri, and Vondr\'{a}k \cite{CCV09} viewed these problems as packing bases of a polymatroid and provided a unified perspective. 
Motivated by applications in wireless networks, recent works have studied the problem of packing set covers in the \emph{online} model \cite{PBV15, EGK19, BBJ24}. 
The online model poses new challenges for packing problems. In particular, it is not clear how to pack a maximum number of disjoint spanning trees in a graph when edges arrive online.  Motivated by these applications and theoretical considerations we consider the problem of packing disjoint bases of a polymatroid in the online model. 

A \emph{polymatroid} $f:2^{\mathcal{N}}\rightarrow \mathbb{Z}_{\geq 0}$ on a ground set $\mathcal{N}$ is an integer-valued monotone submodular function with $f(\emptyset)=0$. We recall that a function $f: 2^{\mathcal{N}}\rightarrow \mathbb{Z}_{\geq 0}$ is \emph{monotone} if $f(A)\leq f(B)$ for every $A\subseteq B$ and \emph{submodular} if $f(A)+f(B)\geq f(A\cup B)+f(A\cap B)$ for every $A,B\subseteq \mathcal{N}$. A subset $S\subseteq \mathcal{N}$ is a \emph{base} of $f$ if $f(S)=f(\mathcal{N}).$\footnote{In most settings, one would define a set $S$ to be a base if it is \emph{inclusionwise minimal} subject to satisfying $f(S) = f(\mathcal{N})$. This is particularly important in the setting of matroids where all bases have the same cardinality which is not necessarily true in the polymatroidal setting. However, since we are interested in (approximating) the maximum number of disjoint bases we adopt the relaxed definition for simplicity.}
In the \DPB problem, we are given a polymatroid $f:2^{\mathcal{N}}\rightarrow \mathbb{Z}_{\geq 0}$ via an evaluation oracle. The goal is to find a maximum number of disjoint bases. We define 
\[
\opt(f):=\max\{k: \exists \ k \text{ disjoint bases of $f$}\}.
\]
Polymatroids generalize matroid rank function, coverage functions, and many others. Consequently, numerous set packing problems can be cast as special cases of \DPB. We will discuss some of these special cases shortly. 

\paragraph{Online \DPB Model.} 
We formally describe the online model for \DPB (denoted \ODPB).   
We have an underlying polymatroid $f: 2^{\mathcal{N}}\rightarrow \mathbb{Z}_{\geq 0}$ over a large but finite ground set $\calN$. Let $n:=|\calN|$. 
We index the elements of the ground set $\mathcal{N}$ as $e_1, e_2, \ldots, e_n$ where $e_t$ is the element that arrives at time $t$ for every $t\in [n]$. For each $t\in [n]$, we denote $\mathcal{N}_t:=\{e_1, e_2, \ldots, e_t\}$ and let $f_{|\calN_t}:2^{\calN_t}\rightarrow \mathbb{Z}_{\geq 0}$ be the function obtained from $f$ by restricting the ground set to $\calN_t$, that is, $f_{|\calN_t}(A):=f(A)$ for every $A\subseteq \calN_t$. At each timestep $t\in [n]$, the online algorithm has access to the evaluation oracle\footnote{The evaluation oracle of a function $g: 2^V\rightarrow \mathbb{R}$ takes a set $S\subseteq V$ as input and returns $g(S)$.} of the function restricted to the set of elements that have arrived until time $t$, i.e., evaluation oracle of the function $f_{|\calN_t}$, and has to color element $e_t$ irrevocably.  
A color is said to be a \emph{base color} if the set $S$ of elements with that color is a base of $f$. 
The goal of the online algorithm is to maximize the number of base colors. 
We remark that we are implicitly assuming that the elements of $\calN$ is the input sequence. 
The competitive ratio of an online algorithm is the ratio between the base colors
in an optimal offline algorithm and that of the online algorithm. For randomized online algorithms, we will be interested in the expected competitive ratio. 
We assume that the online algorithm has prior knowledge of the function value of the  ground set $\calN$, i.e., $f(\calN)$. This assumption is motivated by applications to be discussed below. 

\paragraph{Applications.} 
Polymatroids generalize coverage functions and matroid rank functions. 
We discuss these two special cases, the associated packing problems, and their online model below. 

\begin{enumerate}
\item In the \DSC problem, the input is a set system over a finite universe. We will alternatively view the set system as a hypergraph $H=(V, E)$ where $V$ corresponds to the universe and the hyperedges in $E$ correspond to sets in the system. The goal is to find a maximum number of disjoint set covers (a subset $A\subseteq E$ of hyperedges is a set cover if $\cup_{e\in A}e = V$). \DSC can be cast as a special case of \DPB by considering the coverage function of the hypergraph as the polymatroid, i.e., by considering the polymatroid $f: 2^E\rightarrow \Z_{\ge 0}$ defined as $f(A):=|\cup_{e\in A}e|$. Coverage functions are fairly general with several prominent special cases--e.g., the domatic number problem is a special case of \DSC \cite{FHKS02}. In the online setting of \DSC (termed \ODSC), the vertex set is known in advance while the hyperedges are revealed in an online fashion. The online algorithm has to color each hyperedge immediately upon arrival irrevocably in order to maximize the number of colors that form a set cover.


\item 
In the \DMB problem, we are given 
evaluation access to a matroid rank function $r:2^\calN\rightarrow \Z_{\ge 0}$ over a ground set $\calN$ (we recall that a matroid rank function $r$ is a polymatroid that additionally satisfies $r(\{e\})\le 1$ for every $e\in \calN$). A subset $S\subseteq \calN$ is a base of the matroid if $r(S)=r(\calN)$. The goal is to find a maximum number disjoint bases of the matroid. Matroid rank function is a polymatroid and hence, \DMB is a special case of \DPB. In the online setting of \DMB (termed \ODMB), the ground set is revealed in an online fashion while the online algorithm has access to the rank function restricted to the set of elements that have arrived so far. The online algorithm has to color each element immediately upon arrival irrevocably in order to maximize the number of base colors. 
\end{enumerate}

Next, we describe three special cases of \DPB, namely \MatrixDB, \DST, and \DCSS. To the best of authors' knowledge, these three problems have not been explored in the online setting. One of the motivations of this work is to understand these three problems in the online setting. 
\begin{enumerate}
\item In the \MatrixDB problem, we are given a matrix $M\in \mathbb{R}^{n\times d}$ of rank $d$ and the goal is to find a maximum number of disjoint spanning subsets of row-vectors---a subset of row vectors is \emph{spanning} if its linear hull is $\mathbb{R}^d$. This is a special case of \DMB where the matroid is the linear matroid defined by $M$ (and consequently, the rank function of a subset of row-vectors is the dimension of the subspace spanned by them). In the online setting of \MatrixDB (termed \OMatrixDB), the rows of the matrix are revealed in an online fashion and the online algorithm has to color each row immediately upon arrival irrevocably in order to maximize the number of spanning colors. 

\item In the \DST problem, we are given a connected graph $G=(V, E)$
and the goal is to find a maximum number of disjoint spanning trees in $G$. 
\DST is a special case of \DMB where the matroid is the graphic matroid defined by $G$ (and consequently, the rank function of a subset $F\subseteq E$ is $|V|-c(V, F)$, where $c(V, F)$ is the number of components in the graph $(V, F)$). In the online setting of \DST (termed \ODST), the vertex set of $G$ is known in advance while the edges of $G$ are revealed in an online fashion and the online algorithm has to color each edge immediately upon arrival irrevocably in order to maximize the number of connected colors---a color is connected if the edges of the color form a connected graph over the vertex set $V$. 

\item In the \DCSS problem, we are given a connected hypergraph $H=(V, E)$
and the goal is to find a maximum number of disjoint connected spanning subhypergraphs in $H$. 
\DCSS is a special case of \DPB where the polymatroid $f:2^E\rightarrow \Z_{\ge 0}$ of interest is defined as $f(A):=|V|-c(V, A)$ for every $A\subseteq E$, where $c(V, A)$ is the number of components in the hypergraph $(V, F)$. \DCSS arises in the context of packing element-disjoint Steiner trees \cite{CS07}. In the online setting of \DCSS (termed \ODCSS), the vertex set of $G$ is known in advance while the hyperedges of $G$ are revealed in an online fashion and the online algorithm has to color each hyperedge immediately upon arrival irrevocably in order to maximize the number of connected colors---a color is connected if the hyperedges of the color form a connected hypergraph over the vertex set $V$. 
\end{enumerate}

\DCSS generalizes \DST (in both offline and online settings): if the input hypergraph is a graph, then the problem corresponds to \DST. 
\DCSS is also closely related to \DSC in the following sense: both problems are defined over hypergraphs---the latter asks for disjoint spanning subhypergraphs (which is equivalent to disjoint set covers) while the former asks for disjoint \emph{connected} spanning subhypergraphs. 
\DCSS generalizes \DSC (in both offline and online settings) via the following approximation-preserving reduction from the latter to the former: given an instance $H=(V, E)$ of \DSC, add a new vertex $r$ and stretch every hyperedge of $H$ to include $r$ to obtain a new hypergraph $H'=(V+r, E')$; the maximum number of disjoint set covers in $H$ is equal to the maximum number of disjoint connected spanning subhypergraphs in $H'$. 

\paragraph{Prior work in the offline setting.}
\DMB is polynomial-time solvable \cite{Edm65}. \DSC and \DCSS are $o(\log |V|)$-inapproximable and $O(\log|V|)$-approximable \cite{FHKS02, CS07}. 
C{\u{a}}linescu, Chekuri, and Vondr\'{a}k \cite{CCV09} introduced \DPB as a unified generalization of \DMB, \DSC, and \DCSS. They designed an $O(\log{f(\calN)})$-approximation for \DPB by showing an approximate min-max relation for $\opt(f)$. Their approximate min-max relation is based on the following minimization problem:
\[
k^*(f):= \min_{A\subseteq \calN:\ f(A)<f(\calN)}\left\lfloor \frac{\sum_{e\in \calN}(f(A+e)-f(A))}{f(\calN)-f(A)}\right\rfloor. 
\]
It is easy to see that $\opt(f)\le k^*(f)$ (e.g., see \cite{CCV09}). 
Moreover, if the polymatroid $f$ is a matroid rank function, then Edmonds \cite{Edm65} showed that $\opt(f)=k^*(f)$. Edmonds' result is constructive and implies a polynomial time algorithm for \DMB. 
However, for coverage functions, it is known that $\opt(f)\le(1+o(1))k^*(f)/\log{f(\calN)}$ \cite{FHKS02}. C{\u{a}}linescu, Chekuri, and Vondr\'{a}k showed that this bound is tight by giving a polynomial-time algorithm to construct $(1-o(1))k^*(f)/\log{f(\calN)}$ disjoint bases (and hence, $ \opt(f)\ge (1-o(1))k^*(f)/\log{f(\calN)}$). 
Their algorithm unifies the approximation algorithms for \DSC \cite{FHKS02} and \DCSS \cite{CS07}. 
We state their algorithm since it will be important for the rest of our discussion:  
The algorithm computes the parameter $k:=\lfloor k^*(f)/(\log{f(\calN)}+\log{\log{f(\calN)}})\rfloor$ (by guessing/binary search) 
and colors each element with a uniformly random color chosen from a color palette of size $k$. C{\u{a}}linescu, Chekuri, and Vondr\'{a}k showed that the expected number of base colors returned by this random coloring algorithm 
is at least $(1-e/\log{f(\calN)})k = (1-o(1))k^*(f)/\log{f(\calN)}$. An alternative algorithm based on a random permutation is also described in \cite{CCV09} that we discuss later. 

\paragraph{Prior work in the online setting.}
\ODSC was introduced and studied by Pananjady, Bagaria, and Vaze \cite{PBV15} driven by applications to sensor networks, supply chain management, and crowd-sourcing platforms. 
In the context of \DSC, the quantity $k^*(f)$ associated with the coverage function $f$ of the hypergraph $H=(V, E)$ has a simple interpretation: it is equal to the min-degree of the hypergraph $H$. Pananjady, Bagaria, and Vaze showed that if the min-degree is known to the online algorithm in advance, then there is an online \emph{deterministic} algorithm with 
competitive ratio $O(\log{n})$, where $n$ is the number of vertices in the input hypergraph; we note that a randomized online algorithm with the same competitive ratio is an easy consequence of the random coloring algorithm of \cite{CCV09} discussed above. 
On the lower bound side, they showed that if min-degree is \emph{not} known in advance, then every online deterministic algorithm for \DSC has competitive ratio $\Omega(n)$. Although this lower bound result seems to suggest that knowing the min-degree of the graph in advance is required to achieve any meaningful competitive ratio, two different results have overcome this seeming technical barrier by empowering the online algorithm in other ways: 
Firstly, Emek, Goldbraikh, and Kantor \cite{EGK19} designed an online randomized algorithm with \emph{expected} competitive ratio $O(\log^2{n})$ (assuming no knowledge of the min-degree but using randomness). On the lower bound side, they showed that every online randomized algorithm has expected \emph{impure} competitive ratio\footnote{An online (randomized) algorithm for \DSC has \emph{impure} competitive ratio $\alpha$ if the (expected) number of set covers in the online algorithm is at least $(\opt(H)/\alpha)-\beta$ for some $\beta>0$ that is a function only of $n$, where $H$ is the input hypergraph, $n$ is the number of vertices in $H$, and $\opt(H)$ is the maximum number of disjoint set covers in $H$. Our work focuses on the case of $\beta=0$, i.e., pure competitive ratio.} $\Omega(\log{n}/\log{\log{n}})$ (even with knowledge of min-degree). 
Secondly, Bienkowski, Byrka, and Je\.{z} \cite{BBJ24} designed an online deterministic algorithm with \emph{impure} competitive ratio $O(\log^2{n})$ (assuming no knowledge of the min-degree but settling for impure competitive ratio). 

For the more general problem of \ODPB, Pananjady, Bagaria, and Vaze \cite{PBV15} observed that if $k^*(f)$ is known in advance, then it is possible to design a randomized online algorithm with expected competitive ratio $O(\log{f(\calN)})$: indeed, the random coloring algorithm of \cite{CCV09} mentioned above can be implemented in the online setting using knowledge of $k^*(f)$ and it will have the stated competitive ratio (via the results of \cite{CCV09}). In this work, we are interested in \ODPB in the setting where $k^*(f)$ is \emph{not known in advance}. 

Our motivations to consider \ODPB are multifold: 
\MatrixDB and \DST are fundamental problems in linear algebra and graph algorithms respectively. 
We note that \MatrixDB in the online arrival model is non-trivial even for $2$-dimensional vectors, i.e., for $d=2$.
\DCSS generalizes \DST and arises in the context of packing element-disjoint Steiner trees \cite{CS07}. 
As mentioned earlier, \DCSS can also be viewed as a generalization of \DSC. 
Although \DSC has been studied in the online model, there has not been any work on \DCSS in the online model. 

\subsection{Our Results}\label{sec:results}
Throughout this work, we will denote the color palette by the set of natural numbers. 
There is a natural greedy algorithm for \ODPB: initialize color $c=1$; at each timestep $t\in [n]$: use color $c$ for element $e_t$ and if the set of elements with color $c$ is a base, then increment $c$. 
For uniform random arrival order, the competitive ratio of this simple online algorithm is $O(\log{f(\calN)})$ (via the results of \cite{CCV09}); this does not seem to have been explicitly noted in prior work. 
For arbitrary arrival order, the competitive ratio of this online algorithm is $k^*(f)$: the algorithm will return at least one base while the maximum number of possible bases is at most $k^*(f)$ (since $\opt(f)\le k^*(f)$). It is known that $k^*(f)$ is the best possible competitive ratio of deterministic online algorithms that do not have prior knowledge of $k^*(f)$ \cite{PBV15}. 
In this work, we are interested in the setting of arbitrary arrival order without prior knowledge of $k^*(f)$. For this setting, we design a randomized online algorithm with expected competitive ratio $O(\log^2 {f(\calN)})$. 

\begin{theorem}\label{main-result:polymatroid}
    For \DPB, there exists a randomized online algorithm with expected competitive ratio $O(\log^2 {f(\calN)})$. 
    The runtime of the algorithm at each timestep $t\in [n]$ is $poly(t, \log{f(\calN)})$.
\end{theorem}

We recall that the best-known approximation factor for \DPB in the offline setting is $O(\log{f(\calN)})$, and hence the competitive ratio of our online algorithm nearly matches that of the best possible offline algorithm.
We discuss the consequences of Theorem \ref{main-result:polymatroid} for the applications. 
Specializing Theorem \ref{main-result:polymatroid} to coverage functions implies a randomized online algorithm for \DSC with expected competitive ratio $O(\log^2{n})$, where $n$ is the number of vertices of the input hypergraph, thus recovering the result of \cite{EGK19}. Specializing Theorem \ref{main-result:polymatroid} to matroid rank functions implies a randomized online algorithm for \DMB with expected competitive ratio $O(\log^2{r})$, where $r$ is the rank of the ground set. Consequently, we obtain a randomized online algorithm with expected competitive ratio 
$O(\log^2{d})$ for \MatrixDB where $d$ is the dimension of the span of the input vectors and 
with expected competitive ratio $O(\log^2{n})$ for \DST where $n$ is the number of vertices of the input graph. Specializing Theorem \ref{main-result:polymatroid} to the polymatroid function that arises in \DCSS implies a randomized online algorithm for \DCSS with expected competitive ratio $O(\log^2{n})$, where $n$ is the number of vertices of the input hypergraph.

Our randomized online algorithm to prove Theorem \ref{main-result:polymatroid} is based on the notion of quotients of a polymatroid. Quotients played a central role in the recent result on polymatroidal sparsification \cite{quanrud2024quotient}. The competitive ratio analysis of our algorithm is based on novel properties of quotients 
which might be of independent interest. We prove Theorem \ref{main-result:polymatroid} in two steps: firstly, we design a randomized online algorithm with competitive ratio $O(\log^2{f(\calN)})$ but as stated it needs to solve an NP-Hard problem. Next, we modify this algorithm to achieve polynomial run-time while achieving the same competitive ratio. For this, we rely on the \emph{strength decomposition} of polymatroids which is computable in polynomial time and show a property connecting the strength decomposition to min-sized quotients. 
Although these algorithms are general and powerful, they are computationally expensive and difficult to interpret for specific applications.
As our second result, we give a very simple and fast online randomized algorithm for \DCSS and \DST that achieves the same competitive ratio. 

\begin{theorem}\label{main-result:hypergaph}
    For \DCSS, there exists a randomized online algorithm with competitive ratio $O(\log^2 n)$, where $n$ is number of vertices in the input hypergraph. 
    The runtime of the algorithm at each timestep $t$ is $O(|e_t|^2)$, where $e_t$ is the hyperedge that arrives at time $t$. In particular, the algorithm for \DST can be implemented to run in constant time at each timestep. 
\end{theorem}


\subsection{Technical Background and Overview of Algorithms}
We recall that we assume $\mathcal{N}$ is the input sequence. Let $f:2^\mathcal{N} \rightarrow \Z_{\ge 0}$ be the polymatroid of interest. Let  $r:=f(\mathcal{N})$ be the function value of bases and $k^*:=k^*(f)$. For the discussion here, we will assume that $f(e)>0$ for every $e\in \calN$, $r\geq 2$, and $k^*=\Omega(\log^2 r)$ since this is the non-trivial case. 
First consider the setting where the online algorithm knows $k^*$ in advance. In this setting, coloring each element uniformly from the palette $[\Theta(k^*/\log r)]$ gives $\Omega(k^*/\log{r})$ base colors in expectation and, consequently achieves $O(\log{r})$-competitive ratio. 

Suppose $k^*$ is not known in advance. Our approach is to \emph{estimate} $k^*$ at each timestep. 
At each timestep $t\in [n]$, suppose that we can 
find a number $q_t$ that is within a $\text{poly}(r)$ factor of $k^*$, i.e., $q_t\in [\frac{k^*}{\text{poly}(r)}, k^*\cdot \text{poly}(r)]$---we 
call such a $q_t$ as a \emph{coarse estimate}.  Let $P_t$ be a uniform random sample from $\{q_t\cdot 2^{-O(\log r)}, \ldots, q_t\cdot 2^{-1}, q_t\cdot 2^0, q_t\cdot 2^{1}, \ldots, q_t\cdot 2^{O(\log r)}\}$; $P_t$ is a $2$-approximation of $k^*$ with probability $\Omega(1/\log r)$. Consequently, 
coloring the element $e_t$ uniformly from the palette $[\Theta(P_t/\log^2 r)]$ leads to $\Omega(k^*/ \log^2 r)$ base colors in 
expectation. 

The key challenge is computing a coarse estimate for $k^*$ which depends on the full input sequence. 
This may not be feasible at \emph{all} timesteps owing to the limited knowledge of the polymatroid.  
Instead, we compute an estimate at each timestep such that the elements at timesteps which achieve the coarse property provide \emph{a sufficiently large value of $k^*$}. We formalize this approach now. 
Suppose we have a subroutine to compute some estimate $q_t$ at each timestep $t\in [n]$. 
We call a timestep $t\in [n]$ to be \emph{good} if the estimate $q_t$ is a coarse estimate of $k^*$ and \emph{bad} otherwise.
We let $\mathcal{N}_{\good}$ be the collection of elements that arrive at good timesteps and focus on the function $f$ restricted to $\calN_{\good}$, i.e., on the function $g:=f_{|\calN_{\good}}$. Suppose that the following two properties hold: (i) Bases of $g$ are also bases of $f$ and (ii) $k^*(g)=\Omega(1) k^*$. 
These two properties suffice to obtain $\Omega(k^*/\log^2{r})$ base colors in expectation via the fine-tuned random coloring argument in the previous paragraph.

We rely on the notion of \emph{quotients} to obtain an estimator satisfying properties (i) and (ii). 
\begin{definition}\label{defn:quotients}
For a polymatroid $h:2^{\mathcal{V}}\rightarrow \mathbb{Z}_{\geq 0}$, a set $Q\subseteq \mathcal{V}$ is a \emph{quotient} of $h$ if $h(e+(\mathcal{V}\setminus Q))>h(\mathcal{V}\setminus Q)$ $\forall e\in Q$, i.e., if each element $e\in Q$ has strictly positive marginal with respect to $\mathcal{V}\setminus Q$. 
\end{definition}

Since the definition of quotients might seem contrived, we interpret it for some concrete polymatroids: 
\begin{enumerate}
\item Let $G=(V, E)$ be a connected hypergraph (graph respectively). Consider the polymatroid $h:2^E\rightarrow \Z_{\ge 0}$ that arises in \DCSS (\DST respectively) with $G$ being the input instance. Quotients of $h$ correspond to union of cut-set of disjoint subsets of vertices, i.e., each quotient $Q\subseteq E$ of $h$ is of the form $Q=\cup_{S\in \mathcal{C}}\delta(S)$, where the family $\mathcal{C}\subseteq 2^V\setminus \{\emptyset, V\}$ is a disjoint family of subsets of vertices. In particular, the minimum size of a non-empty quotient is equal to the global min-cut value of $G$. 
\item Let $G=(V, E)$ be a hypergraph. Consider the polymatroid $h: 2^E\rightarrow \Z_{\ge 0}$ that arises in \DSC with $G$ being the input instance, i.e., $h$ is the coverage function of $G$. Quotients of $h$ correspond to union of vertex-isolating cuts of $G$, i.e., each quotient $Q\subseteq E$ of $h$ is of the form $Q=\cup_{u\in S}\delta(u)$ where  $S\subseteq V$. In particular, the minimum size of a non-empty quotient is equal to the min-degree of $G$. 
\end{enumerate}

We observed that the minimum size of a non-empty quotient of $f$ is an $r$-approximation of $k^*$---see Appendix~\ref{section:appendix-quotient}. This inspired us to use the minimum-sized quotient of $f_{|\calN_t}$ that \emph{contains $e_t$} as the estimate at each timestep $t\in [n]$. That is, we use the estimator defined as 
\begin{align}\label{techniques:estimate-q_t}
    q_t&:=\min \{|Q|: e_t\in Q \subseteq \mathcal{N}_t \text{ and $Q$ is a quotient of $f_{|\mathcal{N}_t}$}\}.
\end{align}

Our main technical contribution is showing that the preceding estimator satisfies the two properties that we discussed earlier.
This result combined with the random coloring process yields the desired $O(\log^2 r)$-competitive ratio.
A technical issue is that the estimator $q_t$ is NP-Hard to compute for general polymatroids (even for matroid rank functions). 
However, we note that a $\text{poly}(r)$-approximation of $q_t$ is sufficient for the desired competitive ratio. We show that an $r$-approximation of $q_t$ can be computed using a \emph{strength decomposition} with respect to $f_{|\mathcal{N}_t}$ which admits a polynomial-time algorithm.

Next, we briefly discuss the idea behind the simpler algorithm for \DCSS and \DST of Theorem \ref{main-result:hypergaph}. We will focus on \DST here. The polymatroid $f:2^E\rightarrow \Z_{\ge 0}$ of interest here is the rank function of the graphic matroid defined by the undirected graph $G=(V, E)$. For this case, the quotient computation
is polynomial-time solvable since $q_t$ is equal to the min $u_t$-$v_t$ cut value in the graph $(V, E_t)$ where $e_t =(u_t,v_t)$ and $E_t = \{e_1, e_2, \ldots, e_t\}$.  However, we show that a much simpler estimator for \DST also achieves properties (i) and (ii) mentioned above: we use 
\[
\eta_t:= \text{number of edges between $u_t$ and $v_t$ in the graph $(V, E_t)$}.
\]
This estimator generalizes also to \DCSS. 

\section{Preliminaries}\label{section:prelim}
Given two integers $a\leq b$, let $[a,b]$ denote the set of integers $x$ with $a\leq x \leq b$ and $[b]$ denote the set of integers $x$ with $1\leq x \leq b$. The $\log(\cdot)$ operator with an unspecified base refers to $\log_2(\cdot)$.
We recall that a polymatroid is an integer-valued monotone submodular function with its value on the empty set being $0$. Let $f:2^\mathcal{N}\rightarrow \mathbb{Z}_{\geq 0}$ be a polymatroid. 
We denote the value of the function $f$ on the ground set $\calN$ by $r:=f(\mathcal{N})$ and recall that a set $S\subseteq \calN$ is a base if $f(S) = r$. 
We will assume that $f(e)>0$ for every element $e\in \calN$: if we have an element $e\in \calN$ with $f(e)=0$, then $f(S) = f(S+e)$ for every $S\subseteq \calN$ by submodularity and consequently, picking an arbitrary color for $e$ does not influence the number of base colors.  

For a subset $A\subseteq \mathcal{N}$, we define the \emph{marginal function with respect to $A$} as the function $f_A:2^\mathcal{N}\rightarrow \mathbb{Z}$ defined by $f_A(S):=f(A\cup S)-f(A)$ for every $S\subseteq \mathcal{N}$. If $f$ is submodular, then the function $f_A$ is also submodular for every $A\subseteq \mathcal{N}$. For ease of notation, for an element $e\in \mathcal{N}$, we let $e$ denote the singleton set $\{e\}$. For every set $S\subseteq \mathcal{N}$, we use $S+e$ and $S-e$ to abbreviate $S\cup \{e\}$ and $S\setminus \{e\}$, respectively.

We need the notion of span, closed sets, and quotients. For a set $S\subseteq \mathcal{N}$, the \emph{span} of $S$ is the set of elements with marginal value $0$ with respect to $S$, i.e., 
    $\text{span}(S):=\{e\in \mathcal{N}:f_S(e)=0\}$. 
We note that for every two sets $A\subseteq B$, $\text{span}(A)\subseteq \text{span}(B)$ since $f$ is a monotone submodular function. A set $S\subseteq \mathcal{N}$ is \emph{closed} if $S=\text{span}(S)$.
A set $Q\subseteq \mathcal{N}$ is a \emph{quotient} of $f$ if $Q=\mathcal{N}\setminus \text{span}(S)$ for some set $S\subseteq \mathcal{N}$. This definition of quotient is equivalent to the one in Definition \ref{defn:quotients}. 
We note that the empty set is closed (since $f(e)>0$ for every element $e\in \mathcal{N}$) which implies that $\mathcal{N}$ is a quotient of $f$. 
The notion of quotients plays a central role in polymatroid sparsification \cite{quanrud2024quotient}. 

For a set $S\subseteq \mathcal{N}$, let $f_{|S}:2^S\rightarrow \mathbb{Z}_{\geq 0}$ be the function obtained from $f$ by restricting the ground set to $S$, i.e., $f_{|S}(T):=f(T)$ for every $T\subseteq S$. We recall that if $f$ is a polymatroid, then $f_{|S}$ is also a polymatroid for every $S\subseteq \calN$. Restricting the ground set preserves quotients as shown in the following lemma.
\begin{lemma}\cite{quanrud2024quotient}\label{lemma:restrict-quotient}
    Let $f:2^\mathcal{N}\rightarrow \mathbb{Z}_{\geq 0}$ be a polymatroid. Let $Q$ be a quotient of $f$ and $S\subseteq\mathcal{N}$. Then, $Q\cap S$ is a quotient of $f_{|S}$.
\end{lemma}

We recall that 
\begin{align}
    \opt(f)&=\max\{k:\ \exists k \text{ disjoint bases of }f\}\ \text{and}\notag\\
    k^*(f)&=\min_{A\subseteq \mathcal{N}: f(A)<f(\mathcal{N})} \left\lfloor \frac{\sum_{e\in \mathcal{N}}f_A(e)}{f(\mathcal{N})-f(A)}\right\rfloor. \label{def:opt-k}
\end{align}
We note that there exists a closed set $A$ that achieves the minimum in the definition of $k^*(f)$: suppose $A$ is a set that achieves the minimum, then $f(\text{span}(A))=f(A)<f(\calN)$ by definition of $\text{span}$ while $\sum_{e\in \mathcal{N}} f_A(e) \geq \sum_{e\in \mathcal{N}}f_{\text{span}(A)}(e)$ by submodularity of function $f$ and consequently, $\lfloor\frac{\sum_{e\in \mathcal{N}}f_A(e)}{f(\mathcal{N})-f(A)}\rfloor\ge \lfloor\frac{\sum_{e\in \mathcal{N}}f_{\text{span}(A)}(e)}{f(\mathcal{N})-f(\text{span}(A))}\rfloor$.
If the function $f$ is clear from context, then we denote $\opt(f)$ and $k^*(f)$ by $\opt$ and $k^*$ respectively. 
We need the following approximate min-max relation between $\opt$ and $k^*$. 
\begin{theorem}\cite{CCV09}\label{thm:CCV-k^*}
    $\frac{k^*}{O(\log r)}\leq \opt\leq k^*$.
\end{theorem}


We need the following lemma showing that for every polymatroid $f$, sampling each element independently with probability $\Omega\left(\frac{c\cdot \log r}{k^*}\right)$ gives a base of $f$ with constant probability. A variant of the lemma was shown in \cite{CCV09}. We include a proof of the lemma in the appendix for completeness. 

\begin{restatable}{lemma}{lemmaRandomSamplingGivesBase}\label{lemma:sampling-property}
    Let $f:2^{\mathcal{N}}\rightarrow \mathbb{Z}_{\geq 0}$ be a polymatroid with $f(\mathcal{N})=r\geq 2$. Let $p:=\min\{1, \frac{2\log r}{k^*}\}$ and $S\subseteq \mathcal{N}$ be a subset obtained by picking each element in $\mathcal{N}$ with probability at least $p$ independently at random. Then, $S$ is a base of $f$ with probability at least $\frac{1}{2}$.
\end{restatable}

We recall that the elements of the ground set $\calN$ are indexed as $e_1, e_2, \ldots, e_n$ according to the arrival order, i.e., $e_t$ is the element that arrives at time $t\in [n]$ and moreover, $\calN_t=\{e_1, e_2, \ldots, e_t\}$ is the set of elements that have arrived until time $t$ for every $t\in [n]$. For a coloring of all elements, we say that a fixed color is a \emph{base color} if the set of elements of that color is a base of the polymatroid $f$. 
Let $\texttt{Alg}(f)$ be the number of base colors obtained by an online algorithm $\texttt{Alg}$ for a given polymatroid $f$. A deterministic online algorithm $\texttt{Alg}$ is \emph{(purely) $\alpha$-competitive} if for every polymatroid $f$, we have
\begin{equation}
    \texttt{Alg}(f)\geq \frac{\opt(f)}{\alpha}. \label{def:competitive-ratio}
\end{equation}
A randomized online algorithm is \emph{$\alpha$-competitive} if the bound in (\ref{def:competitive-ratio}) holds in expectation. 
\section{An $O(\log^2 r)$-Competitive Algorithm}\label{section:algorithm}
In this section, we present a randomized online algorithm with expected competitive ratio $O(\log^2 r)$ but the run-time of the algorithm at each timestep will be exponential. The purpose of this section is to illustrate that the online arrival order has sufficient information to achieve a reasonable competitive ratio (albeit in exponential runtime) and to highlight the main ideas underlying the eventual polynomial-time algorithm that will prove Theorem \ref{main-result:polymatroid}. 
We describe the algorithm in Section~\ref{subsection:algorithm}. We present a novel property of the sequence of minimum sized quotients containing the last element with respect to a fixed ordering in Section~\ref{subsection:property-quotient}. We prove the structural property that 
focusing on elements whose estimates are within a $\text{poly}(r)$ factor of $k^*$ does not decrease $k^*$ by more than a constant factor 
in Section~\ref{subsection:structural-property}. We use these properties to analyze the competitive ratio of the algorithm in Section~\ref{subsection:ratio-analysis}. For Sections~\ref{subsection:algorithm}, \ref{subsection:property-quotient}, \ref{subsection:structural-property}, and \ref{subsection:ratio-analysis}, we assume that $r\geq 2$ and $k^*\geq 120\log^2 r$. We discuss how to relax the assumption in Section~\ref{subsection:mixed-algorithm} in which we combine the designed algorithm with algorithms for other cases to obtain an algorithm with competitive ratio $O(\log^2 r)$. 
We will modify the algorithm in Section~\ref{section:poly-algorithm} to design an $O(\log^2 r)$-competitive algorithm that runs in polynomial time at each timestep, thereby completing the proof of Theorem~\ref{main-result:polymatroid}.

\subsection{Algorithm Description}\label{subsection:algorithm}
We assume that $k^*\geq 120\log^2 r$ and $r\geq 2$. The algorithm proceeds as follows: 
At each timestep $t\in [n]$, 
the algorithm computes $$q_t:=\min \{|Q|: e_t\in Q \subseteq \mathcal{N}_t \text{ and $Q$ is a quotient of $f_{|\mathcal{N}_t}$}\}.$$ We note that $\mathcal{N}_t$ is a quotient of $f_{|\mathcal{N}_t}$ with $e_t\in \mathcal{N}_t$ and hence, $q_t$ is well-defined. Let $\ell_t:=\lceil \log q_t\rceil$. Then, the algorithm samples a value $R_t$ from $[\ell_t-3\lceil\log r\rceil, \ell_t+3\lceil\log r\rceil]$ uniformly at random. Finally, the algorithm sets $C(e_t)$ to be an integer picked uniformly at random from $[\lfloor2^{R_t}/(60\cdot \log ^2 r)\rfloor]$. We give a pseudocode of the algorithm in Algorithm~\ref{algo:pseudocode-polymatroid}.

\begin{algorithm2e}[ht]
\caption{Randomized Online Algorithm for \DPB}
\label{algo:pseudocode-polymatroid}
\SetKwInput{KwInput}{Input}                
\SetKwInput{KwOutput}{Output}              
\LinesNumbered

\DontPrintSemicolon
  
    \KwInput{an ordering $e_1, e_2, \ldots, e_n$ of the ground set $\calN$ where $e_t$ is the element that arrives at time $t$ for each $t\in [n]$; $r=f(\mathcal{N})$, and 
    evaluation oracle access to $f_{|\calN_t}$ at each timestep $t\in [n]$, where  $f:2^{\calN}\rightarrow \Z_{\ge 0}$ is a polymatroid. 
    }
    \KwOutput{color assignment $C: \mathcal{N} \rightarrow \mathbb{Z}_+$.}
    
    \SetKwProg{When}{When}{:}{}
    \For{$t=1, 2, \ldots, n$}{
        $q_t\gets \min \{|Q|: e_t\in Q \subseteq \mathcal{N}_t \text{ and $Q$ is a quotient of $f_{|\mathcal{N}_t}$}\}$
        
        $\ell_t\gets \lceil \log q_t \rceil$
        
        Sample $R_t$ $\textit{u.a.r.}$ from $[\ell_t-3\lceil\log r\rceil, \ell_t+3\lceil\log r\rceil]$

        Sample $C(e_t)$ $\textit{u.a.r.}$ from $[\lfloor2^{R_t}/(60\cdot \log ^2 r)\rfloor]$

    }
\end{algorithm2e}

\subsection{Ordered Min-sized Quotients Property}\label{subsection:property-quotient}

In this section, we show 
that there can be at most $r$ distinct arrival times $t$ with the same value of $q_t$.

\begin{lemma}\label{lemma:counting-quotient}
    Let $f: 2^{\calN}\rightarrow \Z_{\ge 0}$ be a polymatroid over the ground set $\calN$ with $f(e)>0$ for every $e\in \mathcal{N}$ and $f(\emptyset)=0$. Let $e_1, e_2, \ldots, e_n$ be an ordering of the ground set $\calN$. For every $t\in [n]$, we define
\[
q_t:=\min\left\{|Q|: e_t\in Q\subseteq \calN_t \text{ and $Q$ is a quotient of $f_{|\calN_t}$}\right\}. 
\]
Then, for every $j\in \Z_+$,
\[
|\left\{t\in [n]: q_t = j\right\}| \le r.
\]

\end{lemma}
\begin{proof}
    We note that $q_t$ is well-defined: empty set is closed since $f(e)>0$ for every $e\in \calN$ and hence, $\calN_t$ is a quotient of $f_{|\calN_t}$ containing $e_t$ for every $t\in [n]$. 
    Let $j\in \Z_+$. 
    Let $t_1<t_2<\ldots< t_{\ell}$ be the timesteps $t$ with $q_t=j$. It suffices to show that $\ell\leq r$. For every $i\in [\ell]$, let $Q_{t_i}$ be a minimum sized quotient of $f_{|\mathcal{N}_{t_i}}$ containing element $e_{t_i}$. We note that $|Q_{t_i}|=q_{t_i}=j$ for every $i\in [\ell]$. We define $S:=\bigcup_{i=1}^{\ell}\left(Q_{t_i}\setminus e_{t_i}\right)$.

    We first show that $e_{t_i}\not\in S$ for every $i\in [\ell]$. Suppose there exists an element $e_{t_{i_1}}\in S$. Then, there exists an index $i_2\in [\ell]$ such that $e_{t_{i_1}}\in Q_{t_{i_2}}\setminus e_{t_{i_2}}$. Hence, $e_{t_{i_1}}\in Q_{t_{i_2}}\setminus e_{t_{i_2}} \subseteq \mathcal{N}_{t_{i_2}}$, which implies that $i_1\leq i_2$. Since $e_{t_{i_1}}\not\in Q_{t_{i_1}}\setminus e_{t_{i_1}}$, we have $i_1 \neq i_2$, which shows that $i_1<i_2$. By Lemma~\ref{lemma:restrict-quotient}, $Q_{t_{i_2}}\cap \mathcal{N}_{t_{i_1}}$ is a quotient of $f_{|\mathcal{N}_{t_{i_1}}}$. We note that $|Q_{t_{i_2}}\cap \mathcal{N}_{t_{i_1}}|<|Q_{t_{i_2}}|=j$ since $e_{t_{i_2}}\in Q_{t_{i_2}}$ and $e_{t_{i_2}}\not\in Q_{t_{i_2}}\cap \mathcal{N}_{t_{i_1}}$. Consequently, $Q_{t_{i_2}}\cap \mathcal{N}_{t_{i_1}}$ contradicts the fact that the smallest quotient of $f_{|\mathcal{N}_{t_{i_1}}}$ containing element $e_{t_{i_1}}$ has size exactly $j$. Hence, $e_{t_i}\not\in S$ for every $i\in [\ell]$.

    We now show that $f(\mathcal{N}_{t_i}\setminus S)<f(\mathcal{N}_{t_{i+1}}\setminus S)$ for every $i\in [\ell-1]$. Let $i\in [\ell-1]$. Since $Q_{t_{i+1}}$ is a quotient of $\mathcal{N}_{t_{i+1}}$ containing element $e_{t_{i+1}}$, we have $f_{\mathcal{N}_{t_{i+1}}\setminus Q_{t_{i+1}}}(e_{t_{i+1}})>0$. Since function $f$ is submodular, we have
    $$f(\mathcal{N}_{t_{i+1}}\setminus S)-f(\mathcal{N}_{t_{i+1}-1}\setminus S) = f_{\mathcal{N}_{t_{i+1}-1}\setminus S}(e_{t_{i+1}}) \geq f_{\mathcal{N}_{t_{i+1}}\setminus Q_{t_{i+1}}}(e_{t_{i+1}})>0.$$
    Hence,
    $$f(\mathcal{N}_{t_i}\setminus S)\leq f(\mathcal{N}_{t_{i+1}-1}\setminus S)<f(\mathcal{N}_{t_{i+1}}\setminus S),$$
    where the first inequality is by monotonicity of the function $f$. Thus, we have $f(\mathcal{N}_{t_i}\setminus S)<f(\mathcal{N}_{t_{i+1}}\setminus S)$ for every $i\in [\ell-1]$.

    This implies that $\{f(\mathcal{N}_{t_i}\setminus S)\}_{i=1}^{\ell}$ is a strictly increasing integer sequence. We note that $f(\mathcal{N}_{t_\ell}\setminus S)\leq r$. Also, since $f$ is a monotone function, we have $f(\mathcal{N}_{t_1}\setminus S)\geq f(e_{t_1})>0$. Hence, $\ell\leq r$.
    
\end{proof}

\begin{remark}
    For better understanding, we interpret Lemma~\ref{lemma:counting-quotient} for some concrete polymatroids defined on graphs. Let $G=(V,E)$ be a connected graph with $n=|V|$ vertices and $m=|E|$ edges. Let $e_1, e_2, \ldots, e_m$ be an arbitrary ordering of the edges in $E$. For $t\in [m]$, we let $E_t:=\{e_1, e_2, \ldots, e_t\}$.
    \begin{enumerate}
        \item Lemma~\ref{lemma:counting-quotient} implies that there are at most $n-1$ distinct timesteps $t\in [m]$ with minimum $u_t-v_t$ cut value in the graph $(V, E_t)$ being exactly $j$ for every $j\in Z_+$, where $u_t, v_t$ are the end vertices of $e_t$. This follows by applying the lemma to the rank function $f: 2^E\rightarrow \Z_{\ge 0}$ of the graphic matroid defined by $G$. We recall that quotients of $f$ correspond to union of cut-sets of disjoint subsets of vertices, i.e., each quotient $Q\subseteq E$ is of the form $Q=\cup_{S\in \mathcal{C}}\delta(S)$, where $\mathcal{C}\subseteq 2^V\setminus \{\emptyset, V\}$ is a disjoint family. Hence, $q_t$ is the minimum $u_t-v_t$ cut value in the graph $(V, E_t)$.
        
        \item Lemma~\ref{lemma:counting-quotient} implies that there are at most $m-n+1$ distinct timesteps $t\in [m]$ at which the shortest cycle containing $e_t$ in the graph $(V, E_t)$ has length exactly $j$ for every $j\in Z_+$. This follows by applying the lemma to the rank function $f: 2^E\rightarrow \Z_{\ge 0}$ of the cographic matroid defined by $G$. We recall that quotients of $f$ correspond to union of cycles, i.e., each quotient $Q\subseteq E$ is of the form $Q=\cup_{F\subseteq \mathcal{F}}F$, where $\mathcal{F}$ is a collection of cycles in $G$. Hence, $q_t$ is the length of the shortest cycle containing $e_t$ in the graph $(V, E_t)$. 
    \end{enumerate}
\end{remark}

\subsection{Structural Property}\label{subsection:structural-property}
For each $t \in [n]$, we define the element $e_t$ to be \emph{good} if $\frac{k^*}{2r}< q_t < 2rk^*$ and \emph{bad} otherwise. Let $\mathcal{N}_{\good}$ be the set of good elements. 
The following is the main result of this section. It shows that bases of $f_{|\calN_{\good}}$ are also bases of $f$ and moreover, the quantity $k^*(f_{|\calN_{\good}})$ is at least a constant fraction of $k^*(f)$. 
\begin{restatable}{lemma}{lemmaGoodOpt}\label{lemma:new-opt}
We have that 
\begin{enumerate}
    \item $f(\calN_{\good})=f(\calN)$ and 
    \item $k^*(f_{|\calN_{\good}})\ge \frac{1}{2} k^*(f)$. 
\end{enumerate}
\end{restatable}
We prove Lemma \ref{lemma:new-opt} in a series of steps. 
We note that an element $e_t$ could be bad because of two reasons: either $q_t$ is too large or it is too small. 
We first show that the only reason for certain elements being bad is that their $q_t$ is too small. 

\begin{lemma}\label{lemma:q_t-upper-bound}
    Let $A\subseteq \mathcal{N}$ be a closed set of $f$. Let $T\geq 1$ be the smallest integer such that
    \begin{equation}\label{inequality-for-bad:k^*}
    \sum_{e\in \mathcal{N}_T\setminus A} \left(f(A+e)-f(A)\right) \geq k^*\cdot \left(f(\mathcal{N})-f(A)\right).
    \end{equation}
    Then, for every element $e_t\in \mathcal{N}_T\setminus A$ with $t\in [T]$, we have that $q_t<2rk^*$. 
\end{lemma}
\begin{proof}
    We recall that $k^*=k^*(f)$. By definition of $k^*$, we have
    $$
    \sum_{e\in \mathcal{N}\setminus A} \left(f(A+e)-f(A)\right) \geq k^*\cdot \left(f(\mathcal{N})-f(A)\right).
    $$
    Since $T$ is the smallest integer satisfying inequality~(\ref{inequality-for-bad:k^*}), we have
    \begin{equation}\label{inequality:new-T}
        \sum_{e\in \mathcal{N}_{T-1}\setminus A} \left(f(A+e)-f(A)\right) < k^*\cdot \left(f(\mathcal{N})-f(A)\right).
    \end{equation}
    Hence,
    \begin{align}\label{inequality:contribution-upper-bound}
        \sum_{e\in \mathcal{N}_T\setminus A} \left(f(A+e)-f(A)\right) &= \left(\sum_{e\in \mathcal{N}_{T-1}\setminus A} \left(f(A+e)-f(A)\right)\right)+\left(f(A+e_{T})-f(A)\right) \notag \\
        &< k^*\cdot \left(f(\mathcal{N})-f(A)\right)+\left(f(A+e_{T})-f(A)\right) \ \ \text{(by inequality~(\ref{inequality:new-T}))} \notag \\
        &\leq k^*\cdot \left(f(\mathcal{N})-f(A)\right)+\left(f(\mathcal{N})-f(A)\right) \ \ \text{(since $f(A+e_T)\leq f(\mathcal{N})$)} \notag \\
        &= (k^*+1)\cdot \left(f(\mathcal{N})-f(A)\right).
    \end{align}
    Since $A$ is a closed set, we have $f(A+e)-f(A)\geq 1$ for every element $e\in \mathcal{N}_T\setminus A$. Thus, using inequality~(\ref{inequality:contribution-upper-bound}), we have
    \begin{equation}\label{inequality:size-upper-bound}
        |\mathcal{N}_T\setminus A| < \sum_{e\in \mathcal{N}_T\setminus A} \left(f(A+e)-f(A)\right) \leq (k^*+1)\cdot (f(\mathcal{N})-f(A)) \leq (k^*+1)r <2rk^*.
    \end{equation}
    
    Since $A$ is a closed set, $\mathcal{N}\setminus A$ is a quotient of $f$. By Lemma~\ref{lemma:restrict-quotient}, for every $t\in [n]$, we know that $\mathcal{N}_t\setminus A$ is a quotient of $f_{|\mathcal{N}_t}$. Hence, for every element $e_t\in \mathcal{N}_T\setminus A$ with $t\in [T]$, we have that $q_t\leq |\mathcal{N}_t\setminus A|\leq |\mathcal{N}_T\setminus A|<2rk^*$, where the last inequality is by inequality~(\ref{inequality:size-upper-bound}).
\end{proof}

Next, we show that dropping all bad elements will not decrease the value of $k^*$ by more than a constant factor for all sets. We first show this for closed sets (Lemma \ref{lemma:new-ratio-closed}) and derive it for arbitrary sets as a corollary (Corollary \ref{corollary:new-ratio}). 
\begin{lemma}\label{lemma:new-ratio-closed}
    For every closed set $A\subseteq \mathcal{N}$,
    $$\sum_{e\in \mathcal{N}_{\good}\setminus A} \left(f(A+e)-f(A)\right) \geq \frac{k^*}{2}\cdot \left(f(\mathcal{N})-f(A)\right).$$
\end{lemma}
\begin{proof}
Let $A\subseteq \mathcal{N}$ be a closed set. 
Let $T\geq 1$ be the smallest integer such that inequality \eqref{inequality-for-bad:k^*} holds. 
    Lemma \ref{lemma:q_t-upper-bound} implies that for every $t\in [T]$, element $e_t\in \mathcal{N}_T\setminus A$ is bad if and only if $q_t\leq \frac{k^*}{2r}$. Hence, we have
    \begin{align}\label{inequality:good-elements}
        |\mathcal{N}_T \setminus (A\cup \mathcal{N}_{\good})| &= |\{t\in [T]: e_t\in \mathcal{N}_T\setminus (A\cup \mathcal{N}_{\good})\}| = \left|\left\{t\in [T] : e_t\in\mathcal{N}_T\setminus A \ \& \ q_t\leq \frac{k^*}{2r} \right\}\right| \notag \\
        &\leq \left|\left\{t\in [T] : q_t\leq \frac{k^*}{2r} \right\}\right| = \sum_{j=1}^{\lfloor \frac{k^*}{2r}\rfloor} \left|\{t\in [T] : q_t=j\}\right| \notag \\
        &\leq \left(\frac{k^*}{2r}\right) r \ \ \text{(by Lemma~\ref{lemma:counting-quotient})} \notag \\
        &= \frac{k^*}{2}.
    \end{align}
    
    Thus,
    $$\begin{aligned}
        &\sum_{e\in \mathcal{N}_{\good}\setminus A} \left(f(A+e)-f(A)\right) \\
        &\quad \quad \quad \geq \sum_{e\in (\mathcal{N}_T \cap\mathcal{N}_{\good})\setminus A} \left(f(A+e)-f(A)\right) \ \ \text{(since $f(A+e)\geq f(A)$ for every $e\in \mathcal{N}$)}\\
        &\quad \quad \quad = \sum_{e\in \mathcal{N}_T\setminus A} \left(f(A+e)-f(A)\right) - \sum_{e\in \mathcal{N}_T \setminus (A\cup \mathcal{N}_{\good})} \left(f(A+e)-f(A)\right)\\
        &\quad \quad \quad \geq k^*\cdot \left(f(\mathcal{N})-f(A)\right) - \sum_{e\in \mathcal{N}_T \setminus (A\cup \mathcal{N}_{\good})} \left(f(A+e)-f(A)\right) \ \ \text{(by inequality~(\ref{inequality-for-bad:k^*}))}\\
        &\quad \quad \quad \geq k^*\cdot \left(f(\mathcal{N})-f(A)\right) - \sum_{e\in \mathcal{N}_T \setminus (A\cup \mathcal{N}_{\good})} \left(f(\mathcal{N})-f(A)\right) \ \ \text{(by monotonicity of $f$)}\\
        &\quad \quad \quad \geq k^*\cdot \left(f(\mathcal{N})-f(A)\right) - \frac{k^*}{2}\cdot \left(f(\mathcal{N})-f(A)\right) \ \ \text{(by inequality~(\ref{inequality:good-elements}))}\\
        &\quad \quad \quad = \frac{k^*}{2}\cdot \left(f(\mathcal{N})-f(A)\right).
    \end{aligned}$$
\end{proof}

\begin{corollary}\label{corollary:new-ratio}
    For every set $A\subseteq \mathcal{N}$,
    $$\sum_{e\in \mathcal{N}_{\good}\setminus A} \left(f(A+e)-f(A)\right) \geq \frac{k^*}{2}\cdot \left(f(\mathcal{N})-f(A)\right).$$
\end{corollary}
\begin{proof}
    We have
    $$\begin{aligned}
        \sum_{e\in \mathcal{N}_{\good}\setminus A} (f(A+e)-f(A)) &\geq \sum_{e\in \mathcal{N}_{\good}\setminus \text{span}(A)} (f(A+e)-f(A)) \ \ \text{(since $A\subseteq \text{span}(A)$)}\\
        &\geq \sum_{e\in \mathcal{N}_{\good}\setminus \text{span}(A)} (f(\text{span}(A)+e)-f(\text{span}(A)) \\ 
        & \ \ \ \ \ \ \ \ \ \ \ \ \ \ \ \ \ \ \text{(since $f_A(e)\geq f_{\text{span}(A)}(e)$ by submodularity of $f$)} \\
        &\geq \frac{k^*}{2} \cdot (f(\mathcal{N})-f(\text{span}(A))) \ \ \text{(by applying Lemma~\ref{lemma:new-ratio-closed} for $\text{span}(A)$)}\\
        &= \frac{k^*}{2}\cdot (f(\mathcal{N})-f(A)). \ \ \text{(since $f(A)=f(\text{span}(A))$)}
    \end{aligned}$$
\end{proof}

We now prove Lemma~\ref{lemma:new-opt}.
\begin{proof}
We recall that $k^*=k^*(f)>0$ by assumption. 
    By applying Corollary~\ref{corollary:new-ratio} for $A=\mathcal{N}_{\good}$, we conclude that $f(\mathcal{N})=f(\mathcal{N}_{\good})$. Hence,
    $$k^*(f_{|\calN{\good}})
    =\min_{A\subseteq \mathcal{N}_{\good}: f(A)<f(\mathcal{N}_{\good})} \frac{\sum_{e\in \mathcal{N}_{\good}}f_A(e)}{f(\mathcal{N}_{\good})-f(A)}
    =\min_{A\subseteq \mathcal{N}_{\good}: f(A)<f(\mathcal{N}_{\good})} \frac{\sum_{e\in \mathcal{N}_{\good}\setminus A}f_A(e)}{f(\mathcal{N})-f(A)},$$
    where the last equality holds since $f_A(e)=0$ for every $e\in A$ and $f(\calN_{\good})=f(\calN)$. By Corollary~\ref{corollary:new-ratio}, for every set $A\subseteq \mathcal{N}$ with $f(A)<f(\calN_{\good})$, we have
    $$\frac{\sum_{e\in \mathcal{N}_{\good}\setminus A}f_A(e)}{f(\mathcal{N})-f(A)}\geq \frac{k^*}{2}.$$
    Hence, $k^*(f_{|\calN{\good}})\geq \frac{1}{2}k^*$.
\end{proof}

\subsection{Competitive Ratio Analysis}\label{subsection:ratio-analysis}
Now, we analyze the competitive ratio of Algorithm~\ref{algo:pseudocode-polymatroid}. 
Let $h:=\lfloor \log k^* \rfloor$. We note that $\frac{k^*}{2}< 2^h \leq k^*$. We now prove that every color $c\in [\lfloor 2^h/ (60\cdot \log^2 r)\rfloor]$ is a base color with constant probability in Algorithm~\ref{algo:pseudocode-polymatroid}.

\begin{lemma}\label{lemma:polymatroid-proper-color-probability}
    Let $c \in [\lfloor 2^h/ (60\cdot \log^2 r)\rfloor]$. Then, $\mathbf{Pr}_{\texttt{Alg$_1$} }[c \text{ is a base color}]\geq \frac{1}{2}$.
\end{lemma}
\begin{proof}
    Let $k^*_{\good}:=k^*(f_{|\calN_{\good}})$. We recall that $k^*=k^*(f)$. 
    By Lemma \ref{lemma:new-opt}, we have that $f(\calN_{\good})=f(\calN)$ and hence, bases of $f_{|\calN_{\good}}$ correspond to bases of $f$. By Lemma \ref{lemma:new-opt}, we also have that  
    $k^*_{\good}\geq \frac{k^*}{2}\geq 2^{h-1}$. Let $e_t\in \mathcal{N}_{\good}$. We have $\frac{k^*}{2r}< q_t < 2rk^*$ and hence, $\lceil \log q_t \rceil -3\lceil \log r\rceil \leq h\leq \lceil \log q_t \rceil+3\lceil \log r \rceil$. Therefore,
    \begin{equation}\label{inequality:probability-analysis}
        \mathbf{Pr}[R_t=h]=\frac{1}{6\lceil \log r \rceil+1}>\frac{1}{15\log r}.
    \end{equation}
    Hence, for every element $e_t\in \mathcal{N}_{\good}$, the probability that $e_t$ is colored with $c$ is
    $$\begin{aligned}
        \mathbf{Pr}[C(e_t)=c] &\geq \mathbf{Pr}[R_t=h]\cdot \mathbf{Pr}[C(e_t)=c|r_t=h] \\
        &= \mathbf{Pr}[R_t=h] \cdot \frac{1}{\lfloor 2^h/ (60\cdot \log^2 r)\rfloor} \\
        &> \frac{1}{15\log r}\cdot \frac{60\cdot \log^2 r}{2^h} \ \ \text{(by inequality~(\ref{inequality:probability-analysis}))}\\
        &\geq \frac{1}{15\log r}\cdot \frac{60\cdot \log^2 r}{2k_{\good}^*}= \frac{2\log r}{k^*_{\good}}.
    \end{aligned}$$
    We recall that $k^*_{\good}=k^*(f_{|\calN_{\good}})$. By applying Lemma~\ref{lemma:sampling-property} for the polymatroid $f_{|\mathcal{N}_{\good}}$, we conclude that the elements in $\mathcal{N}_{\good}$ with color $c$ form a base of $f_{|\mathcal{N}_{\good}}$ with probability at least $\frac{1}{2}$. We recall that bases of $f_{|\calN_{\good}}$ are also bases of $f$. 
    Hence, $c$ is a base color with probability at least $\frac{1}{2}$.
\end{proof}

Lemma \ref{lemma:polymatroid-proper-color-probability} implies a lower bound on the expected number of base colors obtained by Algorithm~\ref{algo:pseudocode-polymatroid}.

\begin{corollary}\label{corollary:number-proper-colors}
    The expected number of base colors obtained by Algorithm~\ref{algo:pseudocode-polymatroid} is at least $$\mathbb{E}[\texttt{Alg$_1$}(f)]\geq \frac{1}{2}\cdot \lfloor 2^h/ (60\cdot \log^2 r)\rfloor.$$
\end{corollary}
\begin{proof}
By Lemma \ref{lemma:polymatroid-proper-color-probability}, we have that 
    $$\begin{aligned}
        \mathbb{E}[\texttt{Alg$_1$}(f)] &= \sum_{c\in [\lfloor 2^h/ (60\cdot \log^2 r)\rfloor]} \mathbf{Pr}_{\texttt{Alg$_1$} }[c \text{ is a base color}] 
        &\ge \frac{1}{2}\cdot \lfloor 2^h/ (60\cdot \log^2 r)\rfloor.
    \end{aligned}$$
\end{proof}


\subsection{Combined Algorithm}\label{subsection:mixed-algorithm}
We recall that our analysis of Algorithm~\ref{algo:pseudocode-polymatroid} assumed that $r\geq 2$ and $k^*\geq 120\log^2 r$. We now combine Algorithm~\ref{algo:pseudocode-polymatroid} with other algorithms to address all ranges of $r$ and $k^*$.

Consider the online algorithm \texttt{Alg$_1^*$} that runs Algorithm~\ref{algo:pseudocode-polymatroid} with probability $\frac{1}{3}$, assigns $C(e_t)=1$ for every element $e_t$ with probability $\frac{1}{3}$, and assigns $C(e_t)=t$ for every element $e_t$ with probability $\frac{1}{3}$. We now show that the resulting online algorithm \texttt{Alg$_1^*$} has competitive ratio $O(\log^2 r)$.

\begin{theorem}\label{theorem:polymatroid-combined-ratio}
    Algorithm \texttt{Alg$_1^*$} has competitive ratio $O(\log^2 r)$.
\end{theorem}
\begin{proof}
    Suppose $r=1$. Then, each singleton forms a base, which implies that $\opt(f)=n$. We recall that the algorithm assigns $C(e_t)=t$ for every element $e_t$ with probability $\frac{1}{3}$. This implies that
    $$\begin{aligned}
        \mathbb{E}[\texttt{Alg$_1^*$}(f)] \geq \frac{1}{3}\cdot \opt(f).
    \end{aligned}$$

    Suppose $r\geq 2$ and $k^*< 120\log^2 r$. Then, the optimum is also smaller than $120\log^2 r$ by Theorem~\ref{thm:CCV-k^*}. We recall that the algorithm assigns $C(e_t)=1$ for every element $e_t$ with probability $\frac{1}{3}$. This implies that
    $$\begin{aligned}
        \mathbb{E}[\texttt{Alg$_1^*$}(f)] \geq \frac{1}{3} > \frac{1}{360\log^2 r}\cdot \opt(f).
    \end{aligned}$$

    Now, we may assume that $k^*\geq 120\log^2 r$ and $r\geq 2$. We recall that the algorithm \texttt{Alg$_1^*$} runs Algorithm~\ref{algo:pseudocode-polymatroid} with probability $\frac{1}{3}$. We have
    \begin{equation}\label{inequality:positive-assumption-k}
    2^h/ (60\cdot \log^2 r)> \frac{1}{60}\cdot \frac{k^*}{2\log^2 r}\geq 1,
    \end{equation}
    which implies that
    \begin{equation}\label{inequality:trunc-integer}
        \lfloor 2^h/ (60\cdot \log^2 r) \rfloor \geq \frac{1}{2}\cdot 2^h/ (60\cdot \log^2 r).
    \end{equation}
    Hence, we have
    $$\begin{aligned}
        \mathbb{E}[\texttt{Alg$_1^*$}(f)]&\geq \frac{1}{3}\cdot \mathbb{E}[\texttt{Alg$_1$}(f)] \\
        &\geq \frac{1}{6}\cdot \lfloor 2^h/ (60\cdot \log^2 r)\rfloor \ \ \text{(by Corollary~\ref{corollary:number-proper-colors})}\\
        &\geq \frac{1}{12}\cdot 2^h/ (60\cdot \log^2 r) \ \ \text{(by inequality~(\ref{inequality:trunc-integer}))}\\
        &\geq \frac{1}{1440\cdot \log^2r}\cdot k^* \\
        &\geq \frac{1}{1440\cdot \log^2r}\cdot \opt(f). \ \ \text{(by Theorem~\ref{thm:CCV-k^*})}
    \end{aligned}$$
\end{proof}

\begin{remark}
    We recall that we assumed $k^*\geq 120\log^2 r$ and $r\geq 2$ in Sections~\ref{subsection:algorithm}, \ref{subsection:property-quotient}, \ref{subsection:structural-property}, and \ref{subsection:ratio-analysis}. The first assumption is to ensure the correctness of inequality~(\ref{inequality:positive-assumption-k}), which is used in the last case of the proof for Theorem~\ref{theorem:polymatroid-combined-ratio}. The second assumption is used for inequality~(\ref{inequality:probability-analysis}) in the proof for Lemma~\ref{lemma:polymatroid-proper-color-probability}.
\end{remark}
\section{Approximation of $q_t$ in Polynomial Time}\label{section:poly-algorithm}
In Section~\ref{section:algorithm}, we presented a randomized online algorithm with competitive ratio $O(\log^2 r)$. The algorithm takes exponential time, since the parameter $q_t=\min \{|Q|: e_t\in Q \subseteq \mathcal{N}_t \text{ and $Q$ is a quotient of $f_{|\mathcal{N}_t}$}\}$ is NP-hard to compute\footnote{The problem of computing $q_t$ for $f$ being a matroid rank function is NP-hard. We note that the problem is equivalent to computing a min-sized cocircuit of a matroid containing a specified element. Computing a min-sized cocircuit of a matroid is a NP-hard problem \cite{Var97} that reduces to finding a min-sized cocircuit of a matroid containing a specified element (an algorithm for the latter problem can be applied for each choice of the element to solve the former problem).}. In this section, we show how to get an approximation of $q_t$ in polynomial time via \emph{strength decomposition} of the function $f$ and complete the proof of Theorem~\ref{main-result:polymatroid}.

\subsection{Strength Decomposition}
We first introduce the definition of strength decomposition. Let $f:2^\mathcal{N}\rightarrow \mathbb{Z}_{\geq 0}$ be a polymatroid over the ground set $\mathcal{N}$. For a subset $T$ of $S$, the \emph{strength-ratio} of $T$ in $S$, denoted $\varphi(T|S)$, is the value
$$\varphi(T|S):=\frac{|S|-|T|}{f(S)-f(T)},$$
with the convention that $x/0=+\infty$ for every $x\geq 0$.

\begin{definition}\cite{quanrud2024quotient} 
Let $f:2^\mathcal{N}\rightarrow \mathbb{Z}_{\geq 0}$ be a polymatroid over the ground set $\mathcal{N}$. 
    A \emph{strength decomposition} of $\mathcal{N}$ with respect to $f$ is a sequence of sets $S_0\supseteq S_1 \supseteq \ldots \supseteq S_w$ such that:
    \begin{enumerate}[label=(\arabic*)]
        \item $S_0=\mathcal{N}$,
        \item $S_w=\emptyset$ and $S_i\neq \emptyset$ for every $i\in [0,w-1]$,
        \item For every $i\in [w]$, $S_i=\mathop{argmin}\limits_{S\subseteq S_{i-1}} \ \varphi(S|S_{i-1})$, and
        \item The strength ratios $\varphi(S_i|S_{i-1})$ are nondecreasing in $i$.
    \end{enumerate}
\end{definition}

For every polymatroid $f:2^{\mathcal{N}}\rightarrow \mathbb{Z}_{\geq 0}$ over a ground set $\calN$, there exists a strength decomposition of $\calN$ with respect to $f$ and it can be computed in polynomial time via reduction to submodular minimization \cite{Nar91, Narayanan-book, Fuj09} (also see \cite{quanrud2024quotient}). We summarize this result below. 

\begin{theorem}\label{theorem:strength-decomposition}
    Given a polymatroid $f:2^{\mathcal{N}}\rightarrow \mathbb{Z}_{\geq 0}$ over ground set $\mathcal{N}$ via an evaluation oracle, a strength decomposition of $\mathcal{N}$ with respect to $f$ exists and can be computed in polynomial time.
\end{theorem}

We show that the function value of sets in a strength decomposition sequence is strictly decreasing.

\begin{lemma}\label{lemma:decomposition-value}
    Let $f:2^\mathcal{N}\rightarrow \mathbb{Z}_{\geq 0}$ be a polymatroid over the ground set $\mathcal{N}$. Let $\calN=S_0\supseteq S_1\supseteq \ldots \supseteq S_w=\emptyset$ be a strength decomposition with respect to $f$. Then, $\{f(S_i)\}_{i=0}^{w}$ is a strictly decreasing integer sequence.
\end{lemma}
\begin{proof}
    Integrality of the sequence follows since $f$ is a polymatroid. We now show that the sequence is strictly decreasing. We fix an index $i\in [w]$ and prove that $f(S_{i-1})>f(S_i)$. Since $S_{i-1}\neq \emptyset$ and $f(e)>0$ for every element $e\in \mathcal{N}$, we have $f(S_{i-1})>0$. We recall that $S_i=\mathop{argmin}\limits_{S\subseteq S_{i-1}} \ \varphi(S|S_{i-1})$, which implies that $\varphi(S_i|S_{i-1})\leq \varphi(\emptyset|S_{i-1})$. Hence,
    \begin{equation}\label{inequality:strength-decomposition}
        \frac{|S_{i-1}|-|S_i|}{f(S_{i-1})-f(S_i)}=\varphi(S_i|S_{i-1})\leq \varphi(\emptyset|S_{i-1})=\frac{|S_{i-1}|}{f(S_{i-1})}<+\infty.
    \end{equation}
    The last inequality is by $f(S_{i-1})>0$. If $f(S_{i-1})=f(S_i)$, then $\frac{|S_{i-1}|-|S_i|}{f(S_{i-1})-f(S_i)}=+\infty$, which contradicts inequality~(\ref{inequality:strength-decomposition}). Thus, $f(S_{i-1})\neq f(S_i)$. We recall that $f$ is a monotone function and $S_i\subseteq S_{i-1}$ by definition of the strength decomposition. Therefore, $f(S_{i-1})>f(S_i)$.
\end{proof}

We use Lemma \ref{lemma:decomposition-value} to conclude that all sets in a strength decomposition are closed. 

\begin{lemma}\label{lemma:decomposition-closed-set}
    Let $f:2^\mathcal{N}\rightarrow \mathbb{Z}_{\geq 0}$ be a polymatroid over the ground set $\mathcal{N}$. Let $\calN=S_0\supseteq S_1\supseteq \ldots \supseteq S_{w}=\emptyset$ be a strength decomposition with respect to $f$. Then, $S_i$ is a closed set for every $i\in [0,w]$.
\end{lemma}
\begin{proof}
    We fix an integer $i\in [0,w]$ and consider the set $S_i$. Let $e\in \mathcal{N}\setminus S_{i}$. Then, there exists an index $j\in [0,i]$ with $e\in S_{j-1}\setminus S_j$. If $f_{S_j}(e)=f(S_j+e)-f(S_j)=0$, then
    $$\begin{aligned}
        \varphi(S_j|S_{j-1}) &= \frac{|S_{j-1}|-|S_j|}{f(S_{j-1})-f(S_j)} \\
        &> \frac{|S_{j-1}|-|S_j+e|}{f(S_{j-1})-f(S_j)} \ \ \text{(since $f(S_{j-1})>f(S_j)$ by Lemma~\ref{lemma:decomposition-value})}\\
        &= \frac{|S_{j-1}|-|S_j+e|}{f(S_{j-1})-f(S_j+e)} \ \ \text{(since $f(S_j)=f(S_j+e)$)} \\
        &= \varphi((S_j+e)|S_{j-1}),
    \end{aligned}$$
    which leads to a contradiction that $S_j=\mathop{argmin}\limits_{S\subseteq S_{j-1}} \ \varphi(S|S_{j-1})$. Hence, $f_{S_j}(e)>0$ for every $e\in \mathcal{N}\setminus S_i$. By submodularity of the function $f$, we have $f_{S_i}(e)\geq f_{S_j}(e)>0$ for every $e\in \mathcal{N}\setminus S_i$, which implies that $\texttt{span}(S_i)=S_i$.
\end{proof}

\subsection{Algorithm Description}
Now, we present our online algorithm with competitive ratio $O(\log^2 r)$ in polynomial time. Let us first assume that $k^*\geq 120\log^2 r$ and $r\geq 2$. When the element $e_t$ arrives, instead of computing $q_t$, the algorithm computes a strength decomposition $\calN_t=S_0\supseteq S_1\supseteq \ldots \supseteq S_{w}=\emptyset$ with respect to $f_{|\mathcal{N}_t}$. Let $i\in [w]$ be the index such that $e_t\in S_{i-1}\setminus S_i$. The algorithm computes parameter $\eta_t:=\frac{|\mathcal{N}_t\setminus S_i|}{f(\mathcal{N}_t)-f(S_i)}$. Let $\ell_t:=\lceil \log \eta_t\rceil$. Then, it samples variable $R_t$ from $[\ell_t-3\lceil\log r\rceil, \ell_t+3\lceil\log r\rceil]$ uniformly at random. Finally, the algorithm sets $C(e_t)$ to be an integer picked uniformly at random from $[\lfloor 2^{R_t}/(60\cdot \log ^2 r)\rfloor]$. We give a pseudocode of the algorithm in Algorithm~\ref{algo:pseudocode-polymatroid-polynomial}.

\begin{algorithm2e}[h]
\caption{Randomized Online Algorithm for \DPB}
\label{algo:pseudocode-polymatroid-polynomial}
\SetKwInput{KwInput}{Input}                
\SetKwInput{KwOutput}{Output}              
\LinesNumbered

\DontPrintSemicolon
  
    \KwInput{an ordering $e_1, e_2, \ldots, e_n$ of the ground set $\calN$ where $e_t$ is the element that arrives at time $t$ for each $t\in [n]$; $r=f(\mathcal{N})$, and 
    evaluation oracle access to $f_{|\calN_t}$ at each timestep $t\in [n]$, where  $f:2^{\calN}\rightarrow \Z_{\ge 0}$ is a polymatroid. 
    }
    \KwOutput{color assignment $C: \mathcal{N} \rightarrow \mathbb{Z}_+$.}
    
    \SetKwProg{When}{When}{:}{}
    \For{$t=1,2, \ldots, n$}{
        compute a strength decomposition $\calN_t=S_0\supseteq S_1\supseteq \ldots \supseteq S_w=\emptyset$ with respect to $f_{|\mathcal{N}_t}$
    
        $\eta_t\gets \frac{|\mathcal{N}_t\setminus S_i|}{f(\mathcal{N}_t)-f(S_i)}$ where $e_t\in S_{i-1}\setminus S_i$
        
        $\ell_t\gets \lceil \log q_t \rceil$
        
        Sample $R_t$ $\textit{u.a.r.}$ from $[\ell_t-3\lceil\log r\rceil, \ell_t+3\lceil\log r\rceil]$

        Sample $C(e_t)$ $\textit{u.a.r.}$ from $[\lfloor 2^{R_t}/(60\cdot \log ^2 r)\rfloor]$

        \Return $C(e_t)$
    }
\end{algorithm2e}

\subsection{Relating $\eta_t$ and $q_t$}
We recall that the estimate used in Algorithm \ref{algo:pseudocode-polymatroid} is $$q_t=\min \{|Q|: e_t\in Q \subseteq \mathcal{N}_t \text{ and $Q$ is a quotient of $f_{|\mathcal{N}_t}$}\}.$$ In contrast, the estimate used in Algorithm \ref{algo:pseudocode-polymatroid-polynomial} is $\eta_t$. 
We show that the parameter $\eta_t$ is an $r$-approximation of $q_t$.

\begin{lemma}\label{lemma:eta-approximation}
    For every $t\in [n]$, $\frac{q_t}{r}\leq \eta_t \leq q_t$.
\end{lemma}
\begin{proof}
    We first show the lower bound on $\eta_t$. Let $\calN_t=S_0\supseteq S_1\supseteq \ldots \supseteq S_w=\emptyset$ be a strength decomposition with respect to $f_{|\mathcal{N}_t}$ computed by Algorithm~\ref{algo:pseudocode-polymatroid-polynomial}. Let $i\in [0,w]$ be the index such that $e_t\in S_{i-1}\setminus S_i$. We recall that $\eta_t=\frac{|\mathcal{N}_t\setminus S_i|}{f(\mathcal{N}_t)-f(S_i)}\geq \frac{|\mathcal{N}_t\setminus S_i|}{r}$ since $f(\mathcal{N}_t)-f(S_i)\leq f(\mathcal{N})\leq r$. By Lemma~\ref{lemma:decomposition-closed-set}, $S_i$ is a closed set. Consequently, $\mathcal{N}_t\setminus S_i$ is a quotient of $f_{|\mathcal{N}_t}$ containing element $e_t$ since $e_t\in S_{i-1}\setminus S_i$. This implies that $q_t\leq |\mathcal{N}_t\setminus S_i|$ and hence, $\eta_t\geq \frac{q_t}{r}$.

    We now show the upper bound on $\eta_t$. Let $Q$ be a minimum sized quotient of $f_{|\mathcal{N}_t}$ containing $e_t$ with $|Q|=q_t$. We recall that
    $$\eta_t=\frac{|\mathcal{N}_t\setminus S_i|}{f(\mathcal{N}_t)-f(S_i)}=\frac{\sum_{j=1}^{i}|S_{j-1}\setminus S_j|}{\sum_{j=1}^{i}\left(f(S_{j-1})-f(S_j)\right)}.$$
    For every $j\in [i]$, $\frac{|S_{j-1}\setminus S_j|}{f(S_{j-1})-f(S_j)}=\varphi(S_j|S_{j-1})\leq \varphi(S_i|S_{i-1})$, where the last inequality is because the strength ratios $\varphi(S_i|S_{i-1})$ in a strength decomposition are nondecreasing in $i$. Therefore, we have
    $$\begin{aligned}
        \eta_t &= \frac{\sum_{j=1}^{i}|S_{j-1}\setminus S_j|}{\sum_{j=1}^{i}\left(f(S_{j-1})-f(S_j)\right)} \\
        &\leq \frac{\sum_{j=1}^{i}\left(\varphi(S_i|S_{i-1})\cdot (f(S_{j-1})-f(S_j))\right)}{\sum_{j=1}^{i}\left(f(S_{j-1})-f(S_j)\right)} \ \ \text{(since $\frac{|S_{j-1}\setminus S_j|}{f(S_{j-1})-f(S_j)}\leq \varphi(S_i|S_{i-1})$)} \\
        &= \varphi(S_i|S_{i-1}) \\
        &\leq \varphi(\left(S_{i-1}\setminus Q\right)|S_{i-1}) \ \ \text{(since $S_i=\mathop{argmin}\limits_{S\subseteq S_{i-1}} \ \varphi(S|S_{i-1})$)} \\
        &= \frac{|S_{i-1}|-|S_{i-1}\setminus Q|}{f(S_{i-1})-f(S_{i-1}\setminus Q)}.
    \end{aligned}$$

    Since $Q$ is a quotient of $f_{|\mathcal{N}_t}$, we have $f_{\mathcal{N}_t\setminus Q}(Q)\geq f_{\mathcal{N}_t\setminus Q}(e_t)>0$. Hence, by submodularity of the function $f$, $f_{S_{i-1}\setminus Q}(Q)\geq f_{\mathcal{N}_t\setminus Q}(Q)>0$ since $S_{i-1}\subseteq \mathcal{N}_t$. Consequently, $f(S_{i-1})-f(S_{i-1}\setminus Q)=f_{S_{i-1}\setminus Q}(Q)\geq 1$. This shows that
    $$\begin{aligned}
        \eta_t \leq \frac{|S_{i-1}|-|S_{i-1}\setminus Q|}{f(S_{i-1})-f(S_{i-1}\setminus Q)} \leq |S_{i-1}|-|S_{i-1}\setminus Q|=|S_{i-1}\cap Q| \leq |Q| = q_t.
    \end{aligned}$$
\end{proof}

\subsection{Competitive Ratio Analysis}
Now, we analyze the competitive ratio of Algorithm~\ref{algo:pseudocode-polymatroid-polynomial}. We recall that $h=\lfloor \log k^* \rfloor$ and $\frac{k^*}{2}< 2^h \leq k^*$. We now prove that every color $c\in [\lfloor 2^h/(60\cdot \log ^2 r)\rfloor]$ is a base color with constant probability.

\begin{lemma}\label{lemma:proper-color-probability-poly}
    Let $c \in [\lfloor 2^h/ (60\cdot \log^2 r)\rfloor]$. Then, $\mathbf{Pr}_{\texttt{Alg$_2$} }[c \text{ is a base color}]\geq \frac{1}{2}$.
\end{lemma}
\begin{proof}
    Let $k^*_{\good}:=k^*(f_{|\calN_{\good}})$. We recall that $k^*=k^*(f)$.
   By Lemma \ref{lemma:new-opt}, we have that $f(\calN_{\good})=f(\calN)$ and hence, bases of $f_{|\calN_{\good}}$ correspond to bases of $f$. By Lemma \ref{lemma:new-opt}, we also have that $k^*_{\good}\geq \frac{k^*}{2}\geq 2^{h-1}$. Let $e_t\in \mathcal{N}_{\good}$. We have $\frac{k^*}{2r}< q_t < 2rk^*$. By Lemma~\ref{lemma:eta-approximation}, we have $\frac{k^*}{2r^2}<\frac{q_t}{r}\leq \eta_t\leq q_t<2rk^*$ and hence, $\lceil \log \eta_t \rceil -3\lceil \log r\rceil \leq h\leq \lceil \log \eta_t \rceil+3\lceil \log r \rceil$, which implies that
    \begin{equation}\label{inequality:probability-analysis-polynomial}
        \mathbf{Pr}[R_t=h]=\frac{1}{6\lceil \log r \rceil+1}>\frac{1}{15\log r}.
    \end{equation}
    Hence, 
    for every element $e_t\in \mathcal{N}_{\good}$, 
    the probability that $e_t$ is colored with $c$ is
    $$\begin{aligned}
        \mathbf{Pr}[C(e_t)=c] &\geq \mathbf{Pr}[R_t=h]\cdot \mathbf{Pr}[C(e_t)=c|R_t=h] \\
        &= \mathbf{Pr}[R_t=h] \cdot \frac{1}{\lfloor 2^h/ (60\cdot \log^2 r)\rfloor} \\
        &> \frac{1}{15\log r}\cdot \frac{60\cdot \log^2 r}{2^h} \ \ \text{(by inequality~(\ref{inequality:probability-analysis-polynomial}))}\\
        &\geq \frac{1}{15\log r}\cdot \frac{60\cdot \log^2 r}{2k^*_{\good}}= \frac{2\log r}{k^*_{\good}}.
    \end{aligned}$$
    By applying Lemma~\ref{lemma:sampling-property} for the polymatroid $f_{|\mathcal{N}_{\good}}$, we conclude that elements in $\calN_{\good}$ with color $c$ form a base of $f_{|\calN_{\good}}$ with  probability at least $\frac{1}{2}$. We recall that bases of $f_{|\calN_{\good}}$ are also bases of $f$. Hence, $c$ is a base color with probability at least $\frac{1}{2}$. 
\end{proof}

The lemma above implies a lower bound on the expected number of base colors obtained by Algorithm~\ref{algo:pseudocode-polymatroid-polynomial}.

\begin{corollary}\label{corollary:number-proper-colors-poly}
    The expected number of base colors obtained by Algorithm~\ref{algo:pseudocode-polymatroid-polynomial} is at least $$\mathbb{E}[\texttt{Alg$_2$}(f)]\geq \frac{1}{2}\cdot \lfloor 2^h/ (60\cdot \log^2 r)\rfloor.$$
\end{corollary}
\begin{proof}
By Lemma \ref{lemma:proper-color-probability-poly}, we have that 
    $$\begin{aligned}
        \mathbb{E}[\texttt{Alg$_2$}(f)] &= \sum_{c\in [\lfloor 2^h/ (60\cdot \log^2 r)\rfloor]} \mathbf{Pr}_{\texttt{Alg$_2$} }[c \text{ is a base color}] 
        &\ge \frac{1}{2}\cdot \lfloor 2^h/ (60\cdot \log^2 r)\rfloor.
    \end{aligned}$$
\end{proof}

\subsection{Combined Algorithm}
We recall that our analysis of Algorithm~\ref{algo:pseudocode-polymatroid-polynomial} assumed that $r\geq 2$ and $k^*\geq 120\log^2 r$. We now combine Algorithm~\ref{algo:pseudocode-polymatroid-polynomial} with other algorithms to address all ranges of $r$ and $k^*$.

Consider the online algorithm \texttt{Alg$_2^*$} that runs Algorithm~\ref{algo:pseudocode-polymatroid-polynomial} with probability $\frac{1}{3}$, assigns $C(e_t)=1$ for every element $e_t$ with probability $\frac{1}{3}$, and assigns $C(e_t)=t$ for every element $e_t$ with probability $\frac{1}{3}$. We now show that the online algorithm \texttt{Alg$_2^*$} has competitive ratio $O(\log^2 r)$.

\begin{theorem}\label{theorem:approximation-poly}
    Algorithm \texttt{Alg$_2^*$} has competitive ratio $O(\log^2 r)$.
\end{theorem}
\begin{proof}
    Suppose $r=1$. Then, each singleton forms a base, which implies that $\opt(f)=n$. We recall that the algorithm assigns $C(e_t)=t$ for every element $e_t$ with probability $\frac{1}{3}$. This implies that
    $$\begin{aligned}
        \mathbb{E}[\texttt{Alg$_2^*$}(f)] \geq \frac{1}{3}\cdot \opt(f).
    \end{aligned}$$

    Suppose $r\geq 2$ and $k^*< 120\log^2 r$. Then, the optimum is also smaller than $120\log^2 r$ by Theorem~\ref{thm:CCV-k^*}. We recall that the algorithm assigns $C(e_t)=1$ for every element $e_t$ with probability $\frac{1}{3}$. This implies that
    $$\begin{aligned}
        \mathbb{E}[\texttt{Alg$_2^*$}(f)] \geq \frac{1}{3} > \frac{1}{360\log^2 r}\cdot \opt(f).
    \end{aligned}$$

    Now, we may assume that $k^*\geq 120\log^2 r$ and $r\geq 2$. We recall that the algorithm \texttt{Alg$^*$} runs Algorithm~\ref{algo:pseudocode-polymatroid-polynomial} with probability $\frac{1}{3}$. We have
    $$2^h/ (60\cdot \log^2 r)> \frac{1}{60}\cdot \frac{k^*}{2\log^2 r}\geq 1,$$ which implies that
    \begin{equation}\label{inequality:trunc-integer-poly}
        \lfloor 2^h/ (60\cdot \log^2 r)\rfloor \geq \frac{1}{2}\cdot 2^h/ (60\cdot \log^2 r).
    \end{equation}
    Hence, we have
    $$\begin{aligned}
        \mathbb{E}[\texttt{Alg$_2^*$}(f)]&\geq \frac{1}{3}\cdot \mathbb{E}[\texttt{Alg$_2$}(f)] \\
        &\geq \frac{1}{6}\cdot \lfloor 2^h/ (60\cdot \log^2 r)\rfloor \ \ \text{(by Corollary~\ref{corollary:number-proper-colors-poly})}\\
        &\geq \frac{1}{12}\cdot 2^h/ (60\cdot \log^2 r) \ \ \text{(by inequality~(\ref{inequality:trunc-integer-poly}))}\\
        &\geq \frac{1}{1440\cdot \log^2r}\cdot k^* \\
        &\geq \frac{1}{1440\cdot \log^2r}\cdot \opt(f). \ \ \text{(by Theorem~\ref{thm:CCV-k^*})}
    \end{aligned}$$
\end{proof}

Theorem~\ref{theorem:approximation-poly} completes the proof of Theorem~\ref{main-result:polymatroid} since Algorithm~\ref{algo:pseudocode-polymatroid-polynomial} and Algorithm \texttt{Alg$_2^*$} can be implemented in polynomial time by Theorem~\ref{theorem:strength-decomposition}.
\section{Simpler Algorithm for Online \DCSS}\label{section:hypergraph-faster-algo}
In this section, we consider \DCSS. The input here is a connected hypergraph $G=(V, E)$ where each hyperedge $e\in E$ is a non-empty subset of $V$. We use $n:=|V|$ and $m:=|E|$ to denote the number of vertices and hyperedges in the hypergraph respectively.
We note that $E$ is a multi-set, i.e., it can contain multiple copies of the same hyperedge. We recall that the hypergraph $G=(V, E)$ can equivalently be represented as a bipartite graph $G'=(V\cup E, E')$, where we have a node for each vertex $u\in V$ and each hyperedge $e\in E$ and an edge between nodes $u\in V$ and $e\in E$ if the hyperedge $e$ contains the vertex $u$. A hypergraph is \emph{connected} if its bipartite representation is connected. A subhypergraph of $G$ is a subset of hyperedges $F\subseteq E$ and it is said to be \emph{spanning} if $\cup_{e\in F}e=V$ and \emph{connected} if $(V, F)$ is connected.
The goal is to find a maximum number of disjoint connected spanning subhypergraphs. We define $\opt(G):=\max\{k: \exists \ k \ \text{disjoint spanning hypergraphs}\}$. We note that \DCSS is a special case of \DPB with the polymatroid $f:2^E\rightarrow \Z_{\ge 0}$ of interest being defined as $f(A):=|V|-c(V,A)$ for every $A\subseteq E$, where $c(V,A)$ is the number of components in the bipartite representation of the hypergraph $(V,F)$.

In the online model for \DCSS, the set $V$ of nodes is known in advance while the hyperedges $E$ arrive in an online fashion. For a color assignment $C: E \rightarrow \mathbb{Z}_+$, we say that a color $c\in \Z_+$ is a \emph{connected spanning color} if the set  of hyperedges with color $c$ is spanning and connected. An online algorithm has to color each hyperedge immediately upon arrival irrevocably in order to maximize the number of connected spanning colors. For an online algorithm $\texttt{Alg}$, we denote the number of connected spanning colors obtained by an algorithm on the hypergraph $G$ by $\texttt{Alg}(G)$. 

Theorem~\ref{main-result:polymatroid} implies a randomized online algorithm for this problem with competitive ratio $O(\log^2 n)$ that runs in polynomial time. We recall that Algorithm~\ref{algo:pseudocode-polymatroid-polynomial} needs to compute a strength decomposition of the current hypergraph upon the arrival of each hyperedge. This step is computationally expensive since strength decompositions are obtained via reductions to submodular minimization (in the context of \DCSS, we need repeated application of max-flow). 
We also recall that Algorithm \ref{algo:pseudocode-polymatroid} needs to compute a minimum $s$-$t$ cut of the current hypergraph upon the arrival of each hyperedge, so we need to solve max-flow at each timestep. 
In this section, we design a simple and fast randomized online algorithm with the same competitive ratio for \DCSS. 

We introduce necessary notations and background in Section~\ref{subsection:hypergraph-prelim}. Then, we design an $O(\log^2 n)$ competitive algorithm under the assumption that $\opt(G)\geq 80\cdot \log^2 n$. We describe the algorithm and prove the structural property in Section~\ref{subsection:algorithm-property}. We prove the structural property that focusing on hyperedges with good estimate does not decrease the minimum cut value by more than a constant factor in Section~\ref{section:hypergraph-structural-property}. We analyze the competitive ratio of the algorithm in Section~\ref{subsection:raio-analysis}. For Sections~\ref{subsection:algorithm-property}, \ref{section:hypergraph-structural-property}, and \ref{subsection:raio-analysis}, we assume that $\opt(G)\geq 80\cdot \log^2 n$. We discuss in Section~\ref{subsection:mixed-algorithm-hypergraph} how to relax the assumption by combining the designed algorithm with algorithms for the other case to complete the proof of Theorem~\ref{main-result:hypergaph}.

\subsection{Preliminaries}\label{subsection:hypergraph-prelim}
Let $G=(V,E)$ be a connected hypergraph. Let $n:=|V|$ and $m:=|E|$. For a subset $S\subseteq V$ of vertices, let $\delta(S)$ denote the set of hyperedges intersecting both $S$ and $V\setminus S$. The \emph{value of a global minimum cut of $G$}, denoted $\lambda(G)$, is the minimum number of hyperedges whose removal disconnects $G$, i.e., $\lambda(G):=\min_{\emptyset\neq S\subsetneq V} |\delta(S)|$. We note that $\opt(G)\leq \lambda(G)$ since for every $\emptyset\neq S\subsetneq V$, every connected spanning color $c$ should contain at least one hyperedge in $\delta(S)$. We recall that \DCSS for the input hypergraph $G$ is a special case of \DPB with the polymatroid $f:2^E\rightarrow \Z_{\ge 0}$ of interest being defined as $f(A):=|V|-c(V,A)$ for every $A\subseteq E$, where $c(V,A)$ is the number of components in the bipartite representation of the hypergraph $(V,F)$. We recall that $f(E)=n-1$ since $G$ is connected. Moreover, bases of $f$ correspond to connected spanning subhypergraphs of $G$. We define the \emph{weak partition connectivity of $G$} as 
\[
k^*_G:=k^*(f)= \min_{A\subseteq E:\ f(A)<f(E)}\left\lfloor \frac{\sum_{e\in E}(f(A+e)-f(A))}{f(E)-f(A)}\right\rfloor. 
\]
Weak partition connectivity was studied in hypergraph literature owing to its applications to hypergraph orientation problems \cite{Frank-book}. The terminology of weak partition connectivity defined in the literature is equivalent to the one defined here. 
By applying Lemma~\ref{lemma:sampling-property} to the polymatroid $f$, we have the following lemma. 

\begin{lemma}\label{lemma:sampling-property-hypergraph}
    Let $G=(V, E)$ be a connected hypergraph, $n:=|V|$, and $p:=\min\left\{1, \frac{2\log n}{k^*_G}\right\}$. 
    Let $H=(V,F)$ be a random subhypergraph of $G$ whose hyperedges $F$ are obtained by picking each hyperedge in $E$ with probability at least $p$ independently at random. Then, $H$ is a connected spanning subhypergraph of $G$ with probability at least $\frac{1}{2}$.
\end{lemma}

We also need the following relationship between weak partition connectivity and global min-cut in hypergraphs: 

\begin{lemma}\cite{CCV09}\label{lemma:CCV-hypergraph}
    $\lambda(G)/2 \leq k^*_G \leq \lambda(G)$.
\end{lemma}


For every integer $t\in [m]$, let $e_t$ be the hyperedge arriving at time $t$. Define $E_t:=\{e_1, \ldots, e_t\}$ as the set of hyperedges that have arrived until time $t$. For a subset $S\subseteq V$ of nodes, let 
\[
E_t(S):=\{e\in E_t: S\subseteq e\}
\]
be the set of hyperedges containing all vertices in $S$ that have arrived until time $t$. Let $\delta_t(S)$ be the set of hyperedges in $\delta(S)$ that have arrived until time $t$.

\subsection{Algorithm Description}\label{subsection:algorithm-property}
We recall that $\opt(G)\geq 80\cdot \log^2 n$. The algorithm works as follows: For every $t\in [m]$, when the hyperedge $e_t$ arrives, the algorithm computes $$\eta_t:=\min_{\{u,v\}\in e_t}|E_t(\{u,v\})|.$$ Equivalently, $\eta_t$ is the min over all pairs of vertices in the current hyperedge of the number of hyperedges containing the pair in the current hypergraph. Let $\ell_t:=\lceil \log \eta_t\rceil$. Then, the algorithm samples variable $R_t$ from $[\ell_t, \ell_t+2\lceil\log n\rceil]$ uniformly at random. Finally, the color of the hyperedge $C(e_t)$ is assigned to be an integer picked uniformly at random from $[\lfloor 2^{R_t}/(40\cdot \log ^2n) \rfloor]$. We give a pseudocode of the algorithm in Algorithm~\ref{algo:pseudocode-hypergraph}. We note that $|E_t(\{u,v\})|$ for every $u,v \in V$ can be maintained in $O(|e_t|^2)$ time in the $t$-th timestep, which implies that Algorithm~\ref{algo:pseudocode-hypergraph} takes $O(|e_t|^2)$ time in the $t$-th timestep.

\begin{algorithm2e}[ht]
\caption{Randomized Online Algorithm for \DCSS}
\label{algo:pseudocode-hypergraph}
\SetKwInput{KwInput}{Input}                
\SetKwInput{KwOutput}{Output}              
\LinesNumbered

\DontPrintSemicolon
  
    \KwInput{a hypergraph $G=(V,E)$ with $n$ vertices and $m$ hyperedges; an ordering $e_1, e_2, \ldots, e_m$ of set $E$ where $e_t$ is the hyperedge that arrives at time $t$ for each $t\in [n]$.}
    \KwOutput{color assignment $C: E \rightarrow \mathbb{Z}_+$.}
    
    \SetKwProg{When}{When}{:}{}
    \When{the hyperedge $e_t$ arrives \ }{
        $\eta_t\gets \min\limits_{\{u,v\}\subseteq e_t}|E_t(\{u,v\})|$
        
        $\ell_t\gets \lceil \log \eta_t \rceil$
        
        Sample $R_t$ $\textit{u.a.r.}$ from $[\ell_t, \ell_t+2\lceil \log n \rceil]$

        Sample $C(e_t)$ $\textit{u.a.r.}$ from $[\lfloor 2^{R_t}/(40\cdot \log ^2n) \rfloor]$

        \Return $C(e_t)$
    }
\end{algorithm2e}

\subsection{Structural Property}\label{section:hypergraph-structural-property}
For every $t\in [m]$, we define a hyperedge $e_t$ to be \emph{good} if $\frac{\lambda(G)}{n^2}\leq \eta_t\leq \lambda(G)$ and \emph{bad} otherwise. Let $E_{\text{good}}\subseteq E$ be the collection of all good hyperedges in $G$. Let $G_{\good}:=(V,E_{\text{good}})$ be the subhypergraph of $G$ containing the good hyperedges. 
We show that $G_{\good}$ is connected and moreover, the minimum cut value of $G_{\good}$ is at least a constant fraction of of the minimum cut value of $G$.

\begin{restatable}{lemma}{lemmaGoodOpthypergraph}\label{lemma:new-opt-hypergraph}
We have that 
\begin{enumerate}
    \item $G_{\good}$ is connected and 
    \item $\lambda(G_{\good})\geq \frac{1}{2}\lambda(G)$. 
\end{enumerate}
\end{restatable}


\begin{proof}
We will show that $\lambda(G_{\good})\ge \frac{1}{2}\lambda(G)$. This would imply that $G_{\good}$ is connected since $$\lambda(G_{\good})\geq \frac{\lambda(G)}{2} \geq \frac{\opt(G)}{2} \geq 40\cdot \log^2n \ge 1.$$
The second inequality above is by Lemma \ref{lemma:CCV-hypergraph} and the third inequality is because $\opt(G)\ge 80\log^2{n}$. 

Let $\emptyset\neq S\subsetneq V$. We will show that $|\delta_{G_{\good}}(S)|\ge \lambda(G)/2$. 
    Let $e_{t_1}, e_{t_2}, \ldots, e_{t_k}$ be the hyperedges in $\delta(S)$ with $t_1<t_2<\ldots<t_k$, where $k=|\delta(S)|\geq \lambda(G)$. We let $F:=\{e_{t_i}:1\leq i \leq \lambda(G)\}$ be the first $\lambda(G)$ hyperedges from $\delta(S)$. We now give an upper bound on the number of bad hyperedges in $F$.
    
    For every $i\in [k]$, we recall that $\eta_{t_i}=\min\limits_{\{u,v\}\subseteq e_{t_i}}|E_{t_i}(\{u,v\})|$, which is at most $i$. Hence, every hyperedge in $F$ has its parameter $\eta_{t_i}\leq i \leq k \leq \lambda(G)$. This implies that every bad hyperedge $e_t\in F$ satisfies $\eta_t<\frac{\lambda(G)}{n^2}$. For every two different vertices $u, v \in V$ and positive integer $\texttt{val}<\frac{\lambda(G)}{n^2}$, we call the pair $(\{u, v\}, \texttt{val})$ a \emph{witness} of a bad hyperedge $e_t\in F$ if $\{u,v\}\subseteq e_t$ and $|E_t(\{u,v\})|=\eta_t=\texttt{val}$. Hence, the total number of possible distinct witnesses is smaller than $$|\{\{u,v\}: u,v\in V\}|\cdot \frac{\lambda(G)}{n^2}\leq \frac{n^2}{2}\cdot \frac{\lambda(G)}{n^2}=\frac{\lambda(G)}{2}.$$

    We note that every bad hyperedge in $F$ has at least one witness. For every two different bad hyperedges $e_{t_i}, e_{t_j} \in F$ with $t_i<t_j$, if they share a same witness $(u,v,\texttt{val})$, then we have 
    $$|E_{t_i}(\{u,v\})|=\texttt{val}=|E_{t_j}(\{u,v\})|,$$
    which leads to a contradiction since $|E_{t_i}(\{u,v\})|$ must be strictly smaller than $|E_{t_j}(\{u,v\})|$ by definition. Thus, no two bad hyperedges can share the same witness, which implies that the number of bad hyperedges in $F$ does not exceed the total number of possible distinct witnesses, which is smaller than $\frac{\lambda(G)}{2}$. Therefore, the number of good hyperedges in $F$ is at least $|F|-\frac{\lambda(G)}{2}=\frac{\lambda(G)}{2}$. This shows that $|\delta_{G_{\good}}(S)|\ge \lambda(G)/2$, which further implies that $\lambda(G_{\good})\ge \frac{1}{2}\lambda(G)$.
    
\end{proof}

\subsection{Competitive Ratio Analysis}\label{subsection:raio-analysis}

Now, we analyze the competitive ratio of Algorithm~\ref{algo:pseudocode-hypergraph}. Let $q:=\lfloor \log \lambda(G)\rfloor$. We note that $\lambda(G)\geq 2^q> \lambda(G)/2$. We now prove that every color $c \in [\lfloor 2^q/(40\cdot \log ^2n) \rfloor]$ is a connected spanning color with constant probability.

\begin{lemma}\label{lemma:hypergraph-proper-color-probability}
    Let $c \in [\lfloor 2^q/(40\cdot \log ^2n) \rfloor]$. Then, $\mathbf{Pr}_{\texttt{Alg$_3$} }[c \text{ is a connected spanning color}]\geq \frac{1}{2}$.
\end{lemma}
\begin{proof}
    We recall that $G_{\good}=(V, E_{\texttt{good}})$ is the subhypergraph of $G$ containing the good hyperedges. By Lemma~\ref{lemma:new-opt-hypergraph}, the hypergraph $G_{\good}$ is connected. Moreover, by Lemma~\ref{lemma:new-opt-hypergraph}, we have that $\lambda(G_{\good})\geq \lambda(G)/2$.
    We recall that $k^*_{G_{\good}}$ is the weak partition connectivity of the hypergraph $G_{\good}$.
    By applying Lemma~\ref{lemma:CCV-hypergraph} for the hypergraph $G_{\good}$, we have $\lambda(G_{\good})/2\leq k_{G_{\good}}^*\leq \lambda(G_{\good})$, which implies that
    $$k_{G_{\good}}^*\geq \frac{\lambda(G_{\good})}{2}\geq \frac{\lambda(G)}{4}\geq \frac{2^q}{4}.$$
    
    Let $e_t$ be a hyperedge in $G_{\good}$, i.e., $e_t\in E_{\good}$. Then, we have $\frac{\lambda(G)}{n^2}\leq \eta_t\leq \lambda(G)$ and hence, $\lceil \log \eta_t \rceil \leq q\leq \lceil \log \eta_t \rceil+2\lceil \log n \rceil$. Therefore,
    \begin{equation}\label{inequality:probability-analysis-hypergraph}
        \mathbf{Pr}[R_t=q]=\frac{1}{2\lceil \log n \rceil+1}>\frac{1}{5\log n}.
    \end{equation}
    Then, for every hyperedge $e_t\in E_{\good}$, the probability that $e_t$ is colored with $c$ is
    $$\begin{aligned}
        \mathbf{Pr}[C(e_t)=c] &\geq \mathbf{Pr}[R_t=q] \cdot \mathbf{Pr}[C(e_t)=c|R_t=q] \\
        &\geq \mathbf{Pr}[R_t=q] \cdot \frac{1}{\lfloor 2^q/(40\cdot \log ^2n) \rfloor} \\
        &> \frac{1}{5\log n}\cdot \frac{40\cdot \log^2n}{2^q} \ \ \text{(by inequality~(\ref{inequality:probability-analysis-hypergraph}))} \\
        &\geq \frac{1}{5\log n}\cdot \frac{10\cdot \log^2 n}{k_{G_{\good}}^*} = \frac{2\log n}{k_{G_{\good}}^*}.
    \end{aligned}$$
    By applying Lemma~\ref{lemma:sampling-property-hypergraph} to the connected hypergraph $G_{\good}$, we obtain that $c$ is a connected spanning color with probability at least $\frac{1}{2}$. 
\end{proof}

Lemma \ref{lemma:hypergraph-proper-color-probability} implies a lower bound on the expected number of connected spanning colors obtained by Algorithm~\ref{algo:pseudocode-hypergraph}.

\begin{corollary}\label{corollary:number-proper-colors-hyperedge}
    The expected number of connected spanning colors obtained by Algorithm~\ref{algo:pseudocode-hypergraph} is at least $$\mathbb{E}[\texttt{Alg$_3$}(G)]\geq \frac{1}{2}\cdot \lfloor 2^q/(40\cdot \log ^2n) \rfloor.$$
\end{corollary}
\begin{proof}
By Lemma~\ref{lemma:hypergraph-proper-color-probability}, we have that 
    $$\begin{aligned}
        \mathbb{E}[\texttt{Alg$_3$}(G)] &= \sum_{c\in [\lfloor 2^q/ (40\cdot \log^2 n)\rfloor]} \mathbf{Pr}_{\texttt{Alg$_3$} }[c \text{ is a connected spanning color}] 
        &\ge \frac{1}{2}\cdot \lfloor 2^q/ (40\cdot \log^2 n)\rfloor.
    \end{aligned}$$
\end{proof}

\subsection{Combined Algorithm}\label{subsection:mixed-algorithm-hypergraph}
We recall that our analysis of Algorithm~\ref{algo:pseudocode-hypergraph} assumed that $\opt(G)\geq 80\cdot \log^2 n$. We now combine Algorithm~\ref{algo:pseudocode-hypergraph} with another algorithm to address all ranges of $\opt(G)$.

Consider the online algorithm \texttt{Alg$_3^*$} that runs Algorithm~\ref{algo:pseudocode-hypergraph} with probability $\frac{1}{2}$ and assigns $C(e)=1$ for every hyperedge $e$ with probability $\frac{1}{2}$. We now show that the online algorithm \texttt{Alg$_3^*$} has competitive ratio $O(\log^2 n)$.

\begin{theorem}
    Algorithm \texttt{Alg$_3^*$} has pure competitive ratio $O(\log^2 n)$.
\end{theorem}
\begin{proof}
    Suppose $\opt(G)< 80\cdot \log^2 n$. We recall that the algorithm assigns $C(e)=1$ for every hyperedge $e$ with probability $\frac{1}{2}$. This implies that
    $$\begin{aligned}
        \mathbb{E}[\texttt{Alg$_3^*$}(G)] \geq \frac{1}{2}\cdot \min\{\opt(G),1\} \geq \frac{1}{160\log^2 n}\cdot \opt(G).
    \end{aligned}$$
    
    Now, we may assume that $\opt(G)\geq 80\cdot \log^2 n$. We recall that the algorithm runs Algorithm~\ref{algo:pseudocode-hypergraph} with probability $\frac{1}{2}$. We have
    $$\lfloor 2^q/(40\cdot \log ^2n) \rfloor\geq \frac{\lambda(G)}{80}\cdot \log^2 n\geq \frac{\opt(G)}{80}\cdot \log^2 n\geq 1,$$ which implies that
    \begin{equation}\label{inequality:trunc-integer-hyperedge}
        \lfloor 2^q/(40\cdot \log ^2n) \rfloor \geq \frac{1}{2}\cdot 2^q/(40\cdot \log ^2n).
    \end{equation}
    Hence, we have
    $$\begin{aligned}
        \mathbb{E}[\texttt{Alg$_3^*$}(G)]&\geq \frac{1}{2}\cdot \mathbb{E}[\texttt{Alg}_3(G)] \\
        &\geq \frac{1}{4}\cdot \lfloor 2^q/(40\cdot \log ^2n)\rfloor \ \ \text{(by Corollary~\ref{corollary:number-proper-colors-hyperedge})}\\
        &\geq \frac{1}{8}\cdot 2^q/(40\cdot \log ^2n) \ \ \text{(by inequality~(\ref{inequality:trunc-integer-hyperedge}))} \\
        &> \frac{1}{640\cdot \log^2n}\cdot \lambda(G) \ \ \text{(since $\lambda(G)<2^{q+1}$)}\\
        &\geq \frac{1}{640\cdot \log^2n}\cdot \opt(G). \ \ \text{(since $\opt(G)\leq \lambda(G)$)}
    \end{aligned}$$
\end{proof}

\section{Conclusion}
Packing bases of a polymatroid generalizes numerous set packing problems including disjoint bases of a matroid, disjoint spanning trees of a graph, disjoint set covers of a set system, and disjoint connected spanning subhypergraphs of a hypergraph. In this work, we introduced an online model for packing bases of a polymatroid and gave an algorithm that achieves a polylogarithmic competitive ratio. Our algorithm leads to a polylogarithmic competitive ratio for all these packing problems in a unified fashion. Our algorithm is based on novel properties of the notion of quotients of polymatroids. For the special cases of disjoint spanning trees of a graph and disjoint connected spanning subhypergraphs of a hypergraph, we gave a simpler and more elegant online algorithm that is also easy to implement while achieving the same polylogarithmic competitive ratio. 

Our work leads to several interesting open questions. We mention three prominent ones: 
Firstly, we recall that \DST (and more generally, \DMB) in the offline setting is polynomial-time solvable via matroidal techniques. Is it possible to design an online algorithm for \DST (and more generally, for \DMB) with constant competitive ratio? 
Secondly, it is known that there is no randomized algorithm with expected impure competitive ratio $o(\log{n}/\log\log{n})$ for \DSC, where $n$ is the size of the ground set. Can we show that there is no randomized algorithm with expected impure competitive ratio $o(\log{f(\calN)})$ or pure competitive ratio $o(\log^2{f(\calN)})$ for \DPB? 
Thirdly, we recall that our online model for \DPB assumes knowledge of $f(\calN)$. Is it possible to achieve polylogarithmic competitive ratio without knowledge of $f(\calN)$? This seems to be open even in the special case of coverage functions (i.e., for online \DSC) \cite{EGK19}. 

\bibliographystyle{abbrv}
\bibliography{references}

\newpage
\appendix
\section{Min-size of a non-empty Quotient and $k^*$}
\label{section:appendix-quotient}
Let $q^*$ be the minimum size of non-empty quotients of $f$, i.e., $$q^*:=\min\{|Q|:\emptyset\neq Q\subseteq \mathcal{N} \text{ is a quotient of $f$}\}.$$
We recall that 
$$
    k^*=k^*(f)=\min_{A\subseteq \mathcal{N}: f(A)<f(\mathcal{N})} \left\lfloor \frac{\sum_{e\in \mathcal{N}}f_A(e)}{f(\mathcal{N})-f(A)}\right\rfloor.
$$
The following lemma shows that $q^*$ is an $r$-approximation of $k^*$.

\begin{lemma}\label{lemma:opt-estimation}
    $\frac{q^*}{r}-1< k^*\leq q^*$.
\end{lemma}
\begin{proof}
    We first show the lower bound on $k^*$. Consider a closed set $A\subseteq \mathcal{N}$ with $f(A)<f(\mathcal{N})$ and $k^*=\left\lfloor\frac{\sum_{e\in \mathcal{N}}f_A(e)}{f(\mathcal{N})-f(A)}\right\rfloor$. Since $A$ is closed, we have that $f_A(e)\geq 1$ for every $e\in \mathcal{N}\setminus A$ and $\mathcal{N}\setminus A$ is a quotient of $f$. Therefore, we have
    $$\begin{aligned}
        k^*=\left\lfloor\frac{\sum_{e\in \mathcal{N}}f_A(e)}{f(\mathcal{N})-f(A)}\right\rfloor
        &> \frac{\sum_{e\in \mathcal{N}}f_A(e)}{f(\mathcal{N})-f(A)} -1 \ \ \text{(since $\lfloor x\rfloor > x-1$)}\\
        & \geq \frac{\sum_{e\in \mathcal{N}}f_A(e)}{r}-1 \ \ \text{(since $f(\mathcal{N})-f(A)\leq f(\mathcal{N})\leq r$)}\\
        & \geq \frac{|\mathcal{N}\setminus A|}{r}-1 \ \ \text{(since $f_A(e)\geq 1$ for every $e\in \mathcal{N}\setminus A$)}\\
        & \geq \frac{q^*}{r}-1. \ \ \text{(since $\mathcal{N}\setminus A$ is a quotient of $f$)}
    \end{aligned}$$

    We now show the upper bound on $k^*$. Consider a minimum sized non-empty quotient $Q$. Let $A:=\mathcal{N}\setminus Q$. Since $A$ is a closed set and $A\neq \mathcal{N}$, we have $f(\mathcal{N})>f(A)$. We note that for every $e\in Q$, $f_{A}(e)=f(A+e)-f(A)\leq f(\mathcal{N})-f(A)$, where the inequality is by monotonicity of the function $f$. Therefore,
    $$\begin{aligned}
        q^*=|Q| & \geq \sum_{e\in Q}\left(\frac{f_A(e)}{f(\mathcal{N})-f(A)}\right) \ \ \text{(since $f_A(e)\leq f(\mathcal{N})-f(A)$)}\\
        & = \frac{\sum_{e\in \mathcal{N}} f_{A}(e)}{f(\mathcal{N})-f(A)} \ \ \text{(since $f_A(e)=0$ for every $e\in A$)}\\
        & \geq k^*. \ \ \text{(by definition of $k^*$)}
    \end{aligned}$$
\end{proof}

\section{Random Sampling gives a Base}\label{sec:random-sampling-to-get-base}
In this section, we prove Lemma \ref{lemma:sampling-property}. A variant of the lemma was shown in \cite{CCV09}. 
\lemmaRandomSamplingGivesBase*
\begin{proof}
    If $k^*\leq 2\log r$, then we have $p=1$, which implies that $S=\mathcal{N}$ which is a base of $f$. Henceforth, we may assume that $k^*>2\log r$ and consequently, $p=\frac{2\log r}{k^*}$. We recall that $n=|\mathcal{N}|$. Let $\sigma=(e_1, e_2, \ldots, e_n)$ be a uniformly random permutation of the elements in $\mathcal{N}$. For every $i\in [0,n]$, we define $\mathcal{N}_i:=\{e_j:j\in [i]\}$ and $S_i:=S\cap \mathcal{N}_i$. We note that $S_0=\mathcal{N}_0=\emptyset$ and $f_{S_0}(\mathcal{N})=r$.

    Let $i\in [n]$. We consider the distribution of element $e_i$ conditioned on $S_{i-1}$. Since $\sigma$ is a uniformly random permutation, $e_i$ can be an arbitrary element in $\mathcal{N}\setminus S_{i-1}$, of which each has the same probability. That is, for every element $e\in \mathcal{N}$,
    $$\begin{aligned}
        \mathbf{Pr}_{\sigma, S}[e_i=e|S_{i-1}] = 
        \begin{cases}
            0, \ \ & e\in S_{i-1} \\
            \frac{1}{|\mathcal{N}|-|S_{i-1}|}. \ \ & e\in \mathcal{N}\setminus S_{i-1}
        \end{cases}
    \end{aligned}$$
    We note that element $e_i$ is picked with probability at least $p$. If element $e_i$ is picked, then $S_i=S_{i-1}+e_i$. Otherwise, $S_i=S_{i-1}$. Hence, we have
    $$\begin{aligned}
        \mathbb{E}_{\sigma, S}[f_{S_{i-1}}(\mathcal{N})-f_{S_i}(\mathcal{N}) | S_{i-1}] &= \sum_{e\in \mathcal{N}\setminus S_{i-1}} \frac{1}{|\mathcal{N}|-|S_{i-1}|} \cdot \mathbb{E}_{\sigma, S}[f_{S_{i-1}}(\mathcal{N})-f_{S_i}(\mathcal{N}) | S_{i-1} \ \text{and} \ e_i=e] \\
        &\geq \sum_{e\in \mathcal{N}\setminus S_{i-1}} \frac{p}{|\mathcal{N}|-|S_{i-1}|} \cdot f_{S_{i-1}}(e) \\
        &\geq \frac{p}{n}\cdot \sum_{e\in \mathcal{N}\setminus S_{i-1}}f_{S_{i-1}}(e) \\
        &\geq \frac{p}{n}\cdot k^*\cdot f_{S_{i-1}}(\mathcal{N}),
    \end{aligned}$$
    where the last inequality is by the definition of $k^*$. This implies that
    $$\begin{aligned}
        \mathbb{E}_{\sigma, S}[f_{S_{i}}(\mathcal{N})] &= \mathbb{E}_{\sigma, S}[f_{S_{i-1}}(\mathcal{N})-\mathbb{E}_{\sigma, S}[f_{\mathcal{S}_{i-1}}(\mathcal{N})-f_{\mathcal{S}_i}(\mathcal{N})| S_{i-1}]] \\
        &\leq \mathbb{E}_{\sigma, S}\left[\left(1-\frac{p}{n}\cdot k^*\right)\cdot f_{S_{i-1}}(\mathcal{N})\right]\\
        &= \left(1-\frac{p}{n}\cdot k^*\right)\cdot \mathbb{E}_{\sigma, S}[f_{S_{i-1}}(\mathcal{N})].
    \end{aligned}$$

    By setting $i=n$, we have
    $$\begin{aligned}
        \mathbb{E}_{\sigma, S}[f_{S_{n}}(\mathcal{N})] &\leq \left(1-\frac{p}{n}\cdot k^*\right)^{n} \cdot \mathbb{E}_{\sigma, S}[f_{S_0}(\mathcal{N})] \\
        &= \left(1-\frac{p}{n}\cdot k^*\right)^{n} \cdot r \\
        &< \exp(-p\cdot k^*)\cdot r = \exp(-2\log r)\cdot r < \frac{1}{2}.
    \end{aligned}$$
    According to Markov's inequality, we have
    $$\begin{aligned}
        \mathbf{Pr}_{\sigma, S}[f_{S_{n}}(\mathcal{N}) \geq 1] \leq \mathbb{E}_{\sigma, S}[f_{S_{n}}(\mathcal{N})] < \frac{1}{2},
    \end{aligned}$$
    which shows that $S$ is a base with probability at least $\frac{1}{2}$.
\end{proof}

\end{document}